\documentclass{article}

\usepackage[a4paper]{geometry}

\usepackage{catchfilebetweentags}

\usepackage{amsmath}
\usepackage{amssymb}
\usepackage{bbold}

\usepackage{xspace}
\usepackage{float}
\usepackage{pifont}

\usepackage{microtype}

\usepackage{colonequals}
\usepackage{multirow}
\usepackage{multicol}
\usepackage{url}
\usepackage{tikz}
\usepackage{xcolor}

\usepackage{import}

\usepackage{mathtools}

\usepackage[noconfig]{ntheorem}

\theoremnumbering{arabic}
\theoremstyle{plain}
\theoremsymbol{\ensuremath{_\Box}}
\theorembodyfont{\itshape}
\theoremheaderfont{\normalfont\bfseries}
\theoremseparator{}
\newtheorem{theorem}{Theorem}[section]
\newtheorem{proposition}[theorem]{Proposition}
\newtheorem{lemma}[theorem]{Lemma}
\newtheorem{corollary}[theorem]{Corollary}

\theorembodyfont{\upshape}

\newtheorem{definition}[theorem]{Definition}

\theoremstyle{nonumberplain}
\theoremheaderfont{\itshape}
\theorembodyfont{\normalfont}
\theoremsymbol{\ensuremath{_\blacksquare}}
\newtheorem{proof}{Proof.}
\qedsymbol{\ensuremath{_\blacksquare}}
\theoremclass{LaTeX}

\usetikzlibrary{positioning}
\usetikzlibrary{arrows}

\renewcommand{\paragraph}[1]{\par\vspace{7pt}\noindent\textbf{#1.}}

\newcommand{\PTIME}{\textnormal{P}}

\newcommand{\NP}{\textnormal{NP}}
\newcommand{\coNP}{\textnormal{coNP}}
\newcommand{\EXP}{\textnormal{EXP}}
\newcommand{\coNEXP}{\textnormal{coNEXP}}
\newcommand{\coTwoNEXP}{\textnormal{co2NEXP}}

\newcommand{\M}{\ensuremath{\mathsf{M}}\xspace}

\newcommand{\GZero}{\ensuremath{\mathsf{G_0}}\xspace}

\newcommand\RBE{\ensuremath{\mathsf{RBE}}\xspace}
\newcommand\RBEs{\ensuremath{\mathsf{RBE}\mathrm{s}}\xspace}

\newcommand{\RBEZero}{\ensuremath{\mathsf{RBE_0}}\xspace}

\newcommand\SORBE{\ensuremath{\mathsf{SORBE}}\xspace}

\newcommand{\ShEx}{\ensuremath{\mathsf{ShEx}}\xspace}
\newcommand{\ShExZero}{\ensuremath{\mathsf{ShEx_0}}\xspace}
\newcommand{\DetShEx}{\ensuremath{\mathsf{DetShEx}}\xspace}
\newcommand{\DetShExZero}{\ensuremath{\mathsf{DetShEx_0}}\xspace}
\newcommand{\DetShExZeroMinus}{\ensuremath{\mathsf{DetShEx_0^\spec{-}}}\xspace}

\newcommand{\kind}{\ensuremath{\mathit{kind}}\xspace}

\newcommand{\F}{\ensuremath{\mathcal{F}}\xspace}

\newcommand{\dom}{\ensuremath{\mathsf{dom}}\xspace}
\newcommand{\ran}{\ensuremath{\mathsf{ran}}\xspace}

\newcommand{\interval}[1]{{\ensuremath{\mathord{\text{\normalfont\fontfamily{lmtt}\selectfont{}#1}}}}}
\newcommand{\NONE}{\interval{0}\xspace}
\newcommand{\ONE}{\interval{1}\xspace}
\newcommand{\MAYBE}{\interval{?}\xspace}
\newcommand{\MANY}{\interval{*}\xspace}
\newcommand{\PLUS}{\interval{+}\xspace}

\newcommand{\spec}[1]{\ensuremath{\text{\fontfamily{ccr}\selectfont #1}}}

\newcommand{\sign}{\ensuremath{\mathsf{sign}}\xspace}

\newcommand{\source}{\ensuremath{\mathsf{source}}\xspace}
\newcommand{\target}{\ensuremath{\mathsf{target}}\xspace}
\newcommand{\lab}{\ensuremath{\mathsf{lab}}\xspace}
\newcommand{\occur}{\ensuremath{\mathsf{occur}}\xspace}
\newcommand{\out}{\ensuremath{\mathsf{out}}\xspace}

\newcommand{\Typing}{\ensuremath{\mathsf{Typing}}\xspace}

\newcommand{\mininflow}{\mathit{min\text{-}inflow}}
\newcommand{\maxinflow}{\mathit{max\text{-}inflow}}

\newcommand{\fin}{\mathit{fin}}

\newcommand{\dotmin}{\mathit{.min}}
\newcommand{\dotmax}{\mathit{.max}}

\newcommand{\shuffle}{\mathbin{|\hspace{-0.1em}|}}
\newcommand{\bigshuffle}{\mathbin{\big|\hspace{-0.1em}\big|}}

\newcommand{\ml}{\mathopen{\{\hspace{-0.2em}|}}
\newcommand{\mr}{\mathclose{|\hspace{-0.2em}\}}}

\newcommand{\dbl}{\mathbin{::}}

\newcommand\C{\ensuremath{\mathcal{C}}\xspace}

\newcommand{\Refine}{\mathit{Refine}}

\newcommand{\minus}{\ensuremath{\mathbin{\setminus}}}

\def\clap#1{\hbox to 0pt{\hss#1\hss}}

\def\mathclap{\mathpalette\mathclapinternal}

\def\mathclapinternal#1#2{%
           \clap{$\mathsurround=0pt#1{#2}$}}

\def\qed {{                
   \parfillskip=0pt        
   \widowpenalty=10000     
   \displaywidowpenalty=10000  
   \finalhyphendemerits=0  
   \leavevmode             
   \unskip                 
   \nobreak                
   \hfil                   
   \penalty50              
   \hskip.2em              
   \null                   
   \hfill                  
   $\square$
   \par}}                  

\floatstyle{ruled}
\newfloat{float@algorithm}{htb}{loa}
\floatname{float@algorithm}{Algorithm}
\newcounter{LineCounter@algorithm} 

\newenvironment{BasicCommands@algorithm}{%
}{%
}

\newenvironment{algorithm*}%
{%
\begin{BasicCommands@algorithm}%
\list{}{\itemindent 0em%
        \listparindent\itemindent
        \rightmargin  \leftmargin}%
\item\relax
}{%
\endlist
\end{BasicCommands@algorithm}%
}

{%
  \newcommand{\ResetLineCounter}{%
    \setcounter{LineCounter@algorithm}{0}%
  }%
  \begin{BasicCommands@algorithm}%
  \begin{float@algorithm}%
    \ResetLineCounter%
}%
{%
  \end{float@algorithm}%
  \end{BasicCommands@algorithm}%
}

\def\clap#1{\hbox to 0pt{\hss#1\hss}}

\def\mathclap{\mathpalette\mathclapinternal}

\def\mathclapinternal#1#2{%
           \clap{$\mathsurround=0pt#1{#2}$}}

\title{Containment of Shape Expression Schemas for RDF}
\author{S\l{}awek Staworko\\
  CRIStAL, INRIA LINKS, CNRS\\
  University of Lille 3, France\\
  \url{slawomir.staworko@inria.fr}
  \and Piotr Wieczorek\\
  Institute of Computer Science\\
  University of Wroc\l{}aw, Poland\\
  \url{piotrek@cs.uni.wroc.pl}
}

\begin{document}

\maketitle


\begin{abstract}
  We study the problem of containment for \emph{shape expression schemas}
  (\ShEx) for RDF graphs. We identify a subclass of \ShEx that has a natural
  graphical representation in the form of \emph{shape graphs} and their
  semantics is captured with a tractable notion of \emph{embedding} of an RDF
  graph in a shape graph. When applied to pairs of shape graphs, an embedding is
  a sufficient condition for containment, and for a practical subclass of
  \emph{deterministic shape graphs}, it is also a necessary one, thus yielding a
  subclass with tractable containment. While for general shape graphs a
  \emph{minimal counter-example} i.e., an instance proving non-containment,
  might be of exponential size, we show that containment is
  \EXP-hard and in \coNEXP. Finally, we show that containment for arbitrary \ShEx is
  \coNEXP-hard and in $\coTwoNEXP^\NP$.
\end{abstract}


\section{Introduction}
\label{sec:introduction}

Although RDF has been originally introduced schema-free, it has since then become a
standalone database format and the need for a schema language has been
identified, with new applications of RDF previously reserved to relational and
semi-structured databases~\cite{AGP09,W3CValidationWorkshopReport}. Recently
introduced, and under continuous development by W3C, \emph{shape expression
  schema} (\ShEx) is a formalism for defining languages of valid RDF
graphs~\cite{GPBSSH15,GPSR15,ShEx2W3C,ShExBook17}. \ShEx defines a set of types with each
type defined with a rule describing the admissible types of the outbound
neighborhood of a node. Inspired by versatility of schema languages for
XML~\cite{BeNeVa04,KlScSu03}, the rules of \ShEx are based on regular
expressions. Take for instance the RDF graph in Figure~\ref{fig:rdf-bug-reports}
storing bug reports together with its shape expression schema.
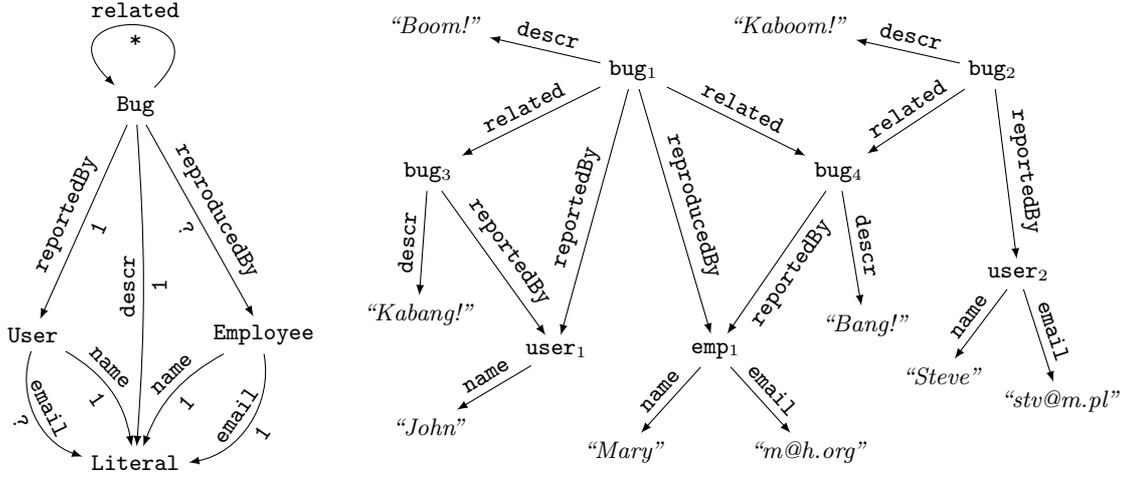
\begin{figure*}[htb]
  \centering
  \small
  \begin{tikzpicture}[>=latex,scale=1.35]
    \path[use as bounding box] (-6.9,-5) rectangle (2,1.8);

    \begin{scope}[xshift=-0.65cm]
    \node (bug1) at (-1,1) {\tt bug${}_1$};
    \node (bug2) at (2.5,1) {\tt bug${}_2$};
    \node (bug3) at (-3,0) {\tt bug${}_3$};
    \node (bug4) at (1,0) {\tt bug${}_4$};
    \node (user1) at (-1.75,-1.75) {\tt user${}_1$};
    \node (user2) at (2.75,-1) {\tt user${}_2$};
    \node (emp1) at (-.2,-1.75) {\tt emp${}_1$};

    \node (d1) at (-2.95,1.45) {\it ``Boom!''};
    \node (d2) at (0.5,1.45) {\it ``Kaboom!''};
    \node (d3) at (-3.1,-1.4) {\it ``Kabang!''};
    \node (d4) at (1.25,-1.5) {\it ``Bang!''};

    \node (n1) at (-3,-2.5) {\it ``John''};
    \node (n2) at (-1.1,-2.75) {\it ``Mary''};
    \node (e2) at (0.7,-2.75) {\it ``m@h.org''};
    \node (n3) at (2,-2) {\it ``Steve''};
    \node (e3) at (3.15,-2.25) {\it ``stv@m.pl''};

    \draw (user1) edge[->] node[above,sloped] {\tt name} (n1);

    \draw (user2) edge[->] node[above,sloped] {\tt name} (n3);
    \draw (user2) edge[->] node[above,sloped] {\tt email} (e3);

    \draw (emp1) edge[->] node[above,sloped] {\tt name} (n2);
    \draw (emp1) edge[->] node[above,sloped] {\tt email} (e2);    

    \draw (bug1) edge[->] node[above,sloped] {\tt related} (bug4);
    \draw (bug1) edge[->] node[above,sloped] {\tt related} (bug3);
    \draw (bug1) edge[->] node[above,sloped] {\tt reportedBy} (user1);
    \draw (bug1) edge[->] node[above,sloped] {\tt reproducedBy} (emp1);
    \draw (bug1) edge[->] node[above,sloped] {\tt descr} (d1);

    \draw (bug2) edge[->] node[above,sloped] {\tt related} (bug4);
    \draw (bug2) edge[->] node[above,sloped] {\tt reportedBy} (user2);
    \draw (bug2) edge[->] node[above,sloped] {\tt descr} (d2);

    \draw (bug3) edge[->] node[above,sloped] {\tt reportedBy} (user1);
    \draw (bug3) edge[->] node[above,sloped] {\tt descr} (d3);

    \draw (bug4) edge[->] node[below,sloped] {\tt reportedBy} (emp1);
    \draw (bug4) edge[->] node[above,sloped] {\tt descr} (d4);
    \end{scope}
  
    \begin{scope}[yshift=0.15cm]
      \node (Bug) at (-6.5, 0.5) {\texttt{Bug}};
      \node (User) at (-7.5, -1.75) {\texttt{User}};
      \node (Emp) at (-5.25, -1.75) {\texttt{Employee}};
      \node (s) at (-6.5, -3) {{\texttt{Literal}}};

      \draw[loop]  (Bug) edge[->] 
                         node[above] {\texttt{related}} 
                         node[below] {\MANY} (Bug); 

      \draw (Bug) edge[->] 
                  node[above,sloped] {\texttt{reportedBy}} 
                  node[below,sloped] {\ONE}
                  (User);
      \draw (Bug) edge[->] 
                  node[above,sloped] {\texttt{reproducedBy}} 
                  node[below,sloped] {\MAYBE}
                  (Emp);

      \draw[bend angle=4] (Bug) edge[->,bend left]
                  node[above,sloped] {\texttt{descr}}
                  node[below,sloped] {\ONE}
                  (s);

      \draw[bend angle=25] (User) edge[->, bend left] 
                   node[above,sloped,pos=0.35] {\texttt{name}} 
                   node[below,sloped] {\ONE}
                   (s);

      \draw[bend angle=45] (User) edge[->, bend right] 
                   node[above,sloped] {\texttt{email}} 
                   node[below,sloped] {\MAYBE}
                   (s);

      \draw[bend angle=20] (Emp) edge[->, bend right] 
                   node[above,sloped,pos=0.4] {\texttt{name}} 
                   node[below,sloped] {\ONE}
                   (s);

      \draw[bend angle=45] (Emp) edge[->, bend left] 
                   node[above,sloped] {\texttt{email}} 
                   node[below,sloped] {\ONE}
                   (s);
                 
    \end{scope}
\node[anchor=center, text width=12.75cm] at (-2.625,-4) {\normalsize
  \begin{align*}
    & \mathtt{Bug} \rightarrow 
    \mathtt{descr}\dbl\mathtt{Literal} ,\ 
    \mathtt{reportedBy}\dbl{}\mathtt{User} ,\ 
    \mathtt{reproducedBy}\dbl{}\mathtt{Employee}^\MAYBE ,\ 
    \mathtt{related}\dbl{}\mathtt{Bug}^\MANY\\
    & \mathtt{User} \rightarrow 
    \mathtt{name}\dbl\mathtt{Literal} ,\  
    \mathtt{email}\dbl{}\mathtt{Literal}^\MAYBE\\
    & \mathtt{Employee} \rightarrow 
    \mathtt{name}\dbl\mathtt{Literal} ,\  
    \mathtt{email}\dbl{}\mathtt{Literal}
  \end{align*}
};

  \end{tikzpicture}
  \caption{An RDF graph with bug reports (top left) together with a
    shape expression schema (bottom) and the corresponding shape graph
    (top right).
  }
  \label{fig:rdf-bug-reports}
\end{figure*}
 The schema requires a bug report to have a
description and a user who reported it. Optionally, a bug report may have an
employee who successfully reproduced the bug. Also, a bug report can have a
number of related bug reports. A user has a name and an optional email address
while an employee has a name and a mandatory email address.

In this paper, we investigate the classical problem of \emph{containment}: given
two schemas $S$ and $S'$, is the set of instances satisfying $S$ contained in
the set of instances satisfying $S'$? This problem has application to a vast
number of problems that perform non-trivial reasoning tasks such as data
exchange, query optimization, or
inference~\cite{APRRS11,CeGoMa15,FeSu98,GoFrMa15,ZhDuYuZh14}. The task at hand
is difficult for a number of reasons.

Because the neighborhood of a node in an RDF graph is unordered, the regular
expression define bag languages, also known as \emph{commutative
  languages}~\cite{HaHo16}, where the relative order among symbols is
irrelevant. This lack of order gives raise to a significant degree of
nondeterminism when working with \emph{regular bag expressions} (\RBE). For
instance, membership for \RBE i.e., deciding whether a bag of symbols belong to
the language defined by an \RBE, is \NP-complete~\cite{KoTo10}. Similarly,
validation for \ShEx i.e., deciding whether a RDF graph satisfies a \ShEx, is
\NP-complete too~\cite{GPBSSH15}. The need for nondeterminism can be limited 
 by disallowing disjunction and permitting the Kleene closure on atomic
symbols only. This yields the class \RBEZero with tractable membership and
tractable validation for the corresponding class of shape expression schemas
\ShExZero. Similarly, single-occurrence regular bag expressions (\SORBE) have
tractable membership and give rise to \emph{deterministic} shape expression
schemas (\DetShEx), where the same symbol can be used only once. Their validation is also tractable.~\cite{GPBSSH15}. Both restrictions offer enough room to
accommodate practical uses, and in particular, the schema in
Figure~\ref{fig:rdf-bug-reports} satisfies them both.

Since \ShEx is a schema language based on types, comparing two schemas requires
the ability to compare types, and consequently, testing $S\subseteq S'$ revolves
around questions whether a type $t$ of $S$ is \emph{covered} by the union of
types $s_1,\ldots,s_k$ of $S'$. Indeed, suppose that in the schema in
Figure~\ref{fig:rdf-bug-reports} we replace the type $\mathtt{User}$ with two
types depending on whether or not the user has an email address:
\begin{small}
\begin{align*}
  & \mathtt{User}_1 \rightarrow 
    \mathtt{name}\dbl\mathtt{Literal}\\[-2pt]
  & \mathtt{User}_2 \rightarrow 
    \mathtt{name}\dbl\mathtt{Literal} ,\  
    \mathtt{email}\dbl{}\mathtt{Literal}
\end{align*}
\end{small}%
and adapt the rest of the schema by replacing $\mathtt{Bug}$ with
\begin{small}
\begin{align*}
  & \mathtt{Bug}_1 \rightarrow 
    \begin{aligned}[t]
      &\mathtt{descr}\dbl\mathtt{Literal} ,\ 
      \mathtt{reportedBy}\dbl{}\mathtt{User}_1 ,\\[-2pt]
      & \smash{\mathtt{reproducedBy}\dbl{}\mathtt{Employee}^\MAYBE} ,\\[-2pt]
      & \smash{\mathtt{related}\dbl{}\mathtt{Bug}_1^\MANY,\ 
      \mathtt{related}\dbl{}\mathtt{Bug}_2^\MANY}
    \end{aligned}\\[-2pt]
  & \mathtt{Bug}_2 \rightarrow 
    \begin{aligned}[t]
      &\mathtt{descr}\dbl\mathtt{Literal} ,\ 
      \mathtt{reportedBy}\dbl{}\mathtt{User}_2 ,\\[-2pt]
      & \smash{\mathtt{reproducedBy}\dbl{}\mathtt{Employee}^\MAYBE} ,\\[-2pt]
      & \smash{\mathtt{related}\dbl{}\mathtt{Bug}_1^\MANY,\ 
      \mathtt{related}\dbl{}\mathtt{Bug}_2^\MANY}
    \end{aligned}
\end{align*}
\end{small}%
Although no longer deterministic (the symbol $\mathtt{related}$ is used twice in
the type definitions of $\mathtt{Bug}_1$ and $\mathtt{Bug}_2$), the modified
schema is equivalent to the original one as the type $\mathtt{Bug}$ is covered
by the union of the types $\mathtt{Bug}_1$ and $\mathtt{Bug}_2$, and the type
$\mathtt{User}$ by the union of $\mathtt{User}_1$ and $\mathtt{User}_2$, the
latter also being equivalent to $\mathtt{Employee}$. Naturally, the fact that a
type might be covered by a union of types is a source of complexity of the
containment problem, and it is an interesting question if there is a class of
schemas for which comparison on pairs of types alone would suffice. For this, we
use \emph{shape graphs}, which are natural graphical representation of \ShExZero
(cf.~Figure~\ref{fig:rdf-bug-reports}), and propose a graph-theoretic notion of
an \emph{embedding} between pairs of shape graphs. In essence, an embedding
identifies in a simulation-like manner when a type is covered by another type,
and therefore, is a sufficient condition for containment. We identify a class
\DetShExZeroMinus for which embedding is also a necessary condition for
containment. Because embeddings are carefully crafted to be tractable, we obtain
a class with tractable containment. The class \DetShExZeroMinus is a subclass of
\DetShExZero, deterministic shape expression schemas that use \RBEZero, with
further minor restrictions that nevertheless do not prevent it from being of
practical use, and in particular, the schema in Figure~\ref{fig:rdf-bug-reports}
belongs to \DetShExZeroMinus. The additional restrictions are necessary as we
show the containment problem for \DetShExZero to be intractable. Interestingly,
for a schema $S$ in \DetShExZeroMinus we construct a \emph{characterizing} graph
$G$ such that $G$ is satisfied by any schema $S'$ in \DetShExZeroMinus if and
only if $S\subseteq S'$.

Checking the containment $S\subseteq S'$ involves implicit negation: checking whether
there is no \emph{counter-example}, an instance that satisfies $S$ and does not
satisfy $S'$. The implicit negation allows to encode disjunction even in classes
of schemas that explicitly forbid using disjunction in the type definition, such
as \ShExZero. This renders \ShExZero very powerful and allows for pairs of
schemas for which the smallest counter-example is of exponential
size. Similarly, we observe a significant increase in complexity: testing
containment for shape graphs is \EXP-hard and in \coNEXP.


The picture of containment for arbitrary shape expression schemas is far from
obvious. It has been observed that $\exists\text{MSO}$ on graphs is alone
insufficient to capture \ShEx and needs to be enriched with Presburger
arithmetic~\cite{SBLGHPS15}. However, monadic extensions of Presburger
arithmetic quickly become undecidable~\cite{ElRa66,SeScMu03}. The question
whether containment for \ShEx is decidable at all is non-trivial and we answer
it positively. While containment for arbitrary shape expressions is
\coNEXP-hard, we show a construction of a counter-example that is triple
exponential but becomes double exponential if a simple compression is used. The
compression is innocuous enough to keep validation of compressed graphs w.r.t.\
\ShEx in \NP. Consequently, we show that the containment for \ShEx is in
$\coTwoNEXP^\NP$. The precise complexity of containment for \ShEx remains an
open question.



Our study has a number of outcomes:
\begin{itemize}
\itemsep0pt
\item a tractable notion of embeddings that is a sufficient condition for
  containment, and a necessary one of a subclass of deterministic shape
  expression schemas.
\item a set of bounds on the size of a minimal graph that satisfies one schema
  and violates another;
\item a thorough characterization of complexity of containment for classes of
  shape expression schemas;
\end{itemize}

\paragraph{Related work}
\emph{tree automata}, \emph{unordered XML schemas}, \emph{implication for graph
  logics}, \emph{containment of \RBEs}, \emph{implication for Presburger
  Arithmetic}
\cite{Ve97}
\url{http://wcms.inf.ed.ac.uk/lfcs/Decidability_Of_Monadic_Theories_Slides.pdf}

\paragraph{Organization}
The paper is organized as follows. In Section~\ref{sec:basics} we present basic
notions. 
In Section~\ref{sec:embeddings} we propose the notion of an embedding and
establish its semantic and computational properties. In
Section~\ref{sec:determinism} we identify a tractable subclass of deterministic
shape expression schemas. In Section~\ref{sec:shex-zero} we analyze containment
for \ShExZero. And in Section~\ref{sec:shex} we investigate the full fragment of
\ShEx. We summarize and outline directions of further study in
Section~\ref{sec:conclusions}. Because of space restriction many proofs are
presented as sketches with full details in appendix.


\section{Basic notions}
\label{sec:basics} 
Throughout this paper we employ elements of function notation to relations, and
conversely, often view functions as relations. For instance, for a binary
relation $R\subseteq A\times B$ we set
$\dom(R)=\{a\in A \mid \exists b\in B.\ (a,b)\in R\}$,
$\ran(B)= \{b\in B \mid \exists a\in A.\ (a,b)\in R\}$,
$R(a) = \{b\in B \mid (a,b)\in R\}$ for $a\in A$, and
$R^{-1}(b)=\{a\in A \mid (a,b)\in R\}$ for $b\in B$.


\paragraph{Intervals}
We use pairs of numbers including the infinite constant $\infty$ to represent
intervals: the pair $[n;m]$, with $n\leq m\leq\infty$, represents the set
$\{i\mid n \leq i \leq m\}$. We assume that both $n$ and $m$ are stored in
binary. We use a number of operators on intervals: every interval $I=[n;m]$ has
its lower bound $\min(I)=n$ and its upper bound $\max(I)=m$. The point-wise
addition of two intervals $A\oplus B=\{a+b\mid a\in A,\ b\in B\}$ has a natural
interpretation: $[n_1;m_1]\oplus[n_2;m_2]=[n_1+n_2;m_1+m_2]$. Note that $[0;0]$
is the neutral element of $\oplus$, and hence, for $k=0$ the expression
$I_1\oplus\ldots\oplus I_k$ evaluates to $[0;0]$. Also
$[n_1;m_1]\subseteq[n_2;m_2]$ iff $n_2\leq n_1\leq m_1\leq m_2$. 

Four \emph{basic intervals} are commonly employed in popular schema languages
for semi-structured databases, listed here together with their shorthand
notation: $\ONE$ stands for $[1;1]$, $\MAYBE$ for $[0;1]$, $\PLUS$ for
$[1;\infty]$, and $\MANY$ for $[0;\infty]$. By $\M$ we denote the set of basic
intervals. Additionally, we shall use $\NONE$ as a short for $[0;0]$.

\paragraph{Bags}
Let $\Delta$ be a finite set of symbols. A \emph{bag} over $\Delta$ is a
function $w:\Delta\rightarrow\mathbb{N}$ that maps a symbol to the number of its
occurrences. The empty bag $\varepsilon$ has $0$ occurrences of every symbol
i.e., $\varepsilon(a)=0$ for every $a\in\Delta$. We present bags using the
notation $\ml a,\ldots \mr$ with elements possibly being repeated. For example,
when $\Delta=\{a,b,c\}$, $w_0=\ml a, a, a, c, c \mr$ represents the function
$w_0(a) = 3$, $w_0(b) = 0$, and $w_0(c) = 2$.

The \emph{bag union} $w_1\uplus w_2$ of two bags $w_1$ and $w_2$ is defined as
$[w_1\uplus w_2](a) = w_1(a)+w_2(a)$ for all $a\in\Delta$.
A \emph{bag language} is a set of bags. The bag union of two languages $L_1$ and
$L_2$ is the language
$L_1\uplus L_2 = \{w_1\uplus w_2\mid w_1\in L_1, w_2\in L_2\}$. Also, for a
given bag language $L$, we define $L^0=\{\varepsilon\}$ and
$L^i=L\uplus L^{i-1}$ for $i>0$.

\paragraph{Regular bag expressions}
\emph{Regular bag expressions} (\RBE) are analogues of regular expressions for
defining bag languages and use disjunction ``$\mid$,'' unordered concatenation
``$\shuffle$,'' and unordered repetition. Formally, they are defined with the
following grammar:
\[
E ::= \epsilon \mid a \mid (E|E) \mid
(E \shuffle E) \mid E^I,
\]
where $a\in\Delta$ and $I$ is an interval. Their semantics is defined as
follows: $L(\epsilon) = \{ \varepsilon \}$, $L(a) = \{\ml a\mr\}$,
$L(E_1\mid E_2) = L(E_1) \cup L(E_2)$,
$L(E_1\shuffle E_2) = L(E_1) \uplus L(E_2)$, and
$L(E^I) = \textstyle\bigcup_{i \in I} L(E)^i$. 
By $\RBEZero$ we denote the class of regular expressions of the form
$a_1^{M_1}\shuffle\ldots\shuffle a_n^{M_n}$, where $a_i\in\Sigma$ and $M_i\in\M$
for $i\in\{1,\ldots,n\}$. We point out that occurrences of symbols need not be
distinct e.g., $a \shuffle a^\PLUS \shuffle b^\MANY$ is \RBEZero.


\paragraph{Graphs} 
We employ a general graph model that allows to capture RDF graphs as well as an
important subclass of shape expression schemas (\ShExZero). Because shape
expressions schemas do not constrain the predicates of the edges of an RDF
graph, we assume a fixed set $\Sigma$ of predicates names used to label edges of
graphs. To allow expressing shape expression schemas, we additionally label each
edge with an occurrence interval, which intuitively indicate the admissible
number of edges of the given kind (cf.~Definition~\ref{def:embedding}). Also,
the general graph model allows multiple edges connecting the same pair of nodes
with the same predicate label, which is not allowed in RDF graphs.
\begin{definition}[General graphs]
  \label{def:graph}
  A \emph{graph} is a tuple $G=(N_G,E_G,\source_G,\target_G,\lab_G,\occur_G)$,
  where $N_G$ is a finite set of \emph{nodes}, $E_G$ is a finite set of
  \emph{edges}, the topology of the graph is defined with the functions
  $\source_G:E_G\rightarrow N_G$ and $\target_G:E_G\rightarrow N_G$ that
  identify respectively the origin node and end point node of an edge,
  $\lab_G:E_G\rightarrow\Sigma$ assigns a (predicate) label to an edge, and
  $\occur_G:E_G\rightarrow\Sigma$ assigns an occurrence interval to an edge. A
  graph is \emph{simple} if it uses only the interval $\ONE$ and has no two edges
  with the same origin, the same end point, and the same label. By $\GZero$ we
  denote the set of all simple graphs. A \emph{shape graph} is a graph that uses
  only basic occurrence intervals (in \M) and we denote the class of all shape
  graphs with \ShExZero. \qed
\end{definition}
For the purposes of studying containment of shape expression schemas the class
of simple graph captures adequately RDF graphs. Indeed, while RDF nodes are
labeled with URIs, literal values, and blank identifiers, and shape expression
schemas can constraint node labels, in general these constraints can be
``simulated.'' For instance, if the schema imposes a type of admissible literal
nodes (integer, date, etc.), literal nodes can be modified to include an
outgoing edge labeled with the type name.

Shape expression schemas constrain the outbound neighborhood of a node, and for
that purpose we identify the set of all outgoing edges of a node $n\in N_G$ with
$\out_G(n)=\{e\in{}E_G\mid\source_G(e) = n\}$. Sometimes, if a node $n$ has an
outgoing edge leading to $m$, we shall say that $n$ is a \emph{parent} of $m$
and $m$ is a \emph{child} of $n$ (even if $n$ and $m$ are the same node). Also,
we call an $a$-edge any edge labeled with $a\in\Sigma$, and analogously, an
$I$-edge any edge with occurrence interval $I$.

\paragraph{Shape Expression Schemas}
Again, we assume a fixed set of predicate labels $\Sigma$. Given a set of type
names $\Gamma$, a \emph{shape expression} over $\Gamma$ is an \RBE over
$\Sigma\times \Gamma$ and in the sequel we write $(a,t)\in\Sigma\times \Gamma$
simply as $a\dbl{}t$. A \emph{shape expression schema} (\ShEx) is a pair
$S=(\Gamma_S,\delta_S)$, where $\Gamma_S$ is a finite set of types, and
$\delta_S$ is a \emph{type definition} function that maps elements of $\Gamma_S$
to shape expressions over $\Gamma_S$. Typically, we present a \ShEx $S$ as a
collection of rules of the form $t\rightarrow E$ to indicate that
$\delta_S(t)=E$, where $E$ is a shape expression (naturally, no two rules shall
have the same left-hand side).
For a class of \RBEs $\C$, by $\ShEx(\C)$ we denote the class of shape
expression schemas using only shape expressions in $\C$.


\begin{figure}[htb]
  \centering
  \begin{tikzpicture}[>=latex, scale=0.85]
    \node at (-0.15, 1.25) {$G_0$:};
    \node[] (n0) at (0,0) {$n_0$};
    \node[] (n1) at (1.5,0) {$n_1$};
    \node[] (n2) at (3,0) {$n_2$};
    \draw (n0) edge[->] node[above] {$a$} (n1);
    \draw (n1) edge[->] node[above] {$c$} (n2);
    \draw (n1) edge[loop,->] node[above] {$b$} (n1);

    \node at (3.2, 1.25) {$S_0$:};
    \node[right] at (3.5, 0.65) {
      $\begin{array}{l}
         t_0 \rightarrow a\dbl t_1\\
         t_1 \rightarrow b\dbl t_2\shuffle c\dbl t_3\\
         t_2 \rightarrow {b\dbl t_2}^\MAYBE \shuffle c\dbl t_3\\
         t_3 \rightarrow \epsilon
       \end{array}$
     };
  \end{tikzpicture}
  \caption{\label{fig:graph-and-schema} A simple graph $G_0$ and a schema
    $S_0$.}
\end{figure}

We recall the formal semantics of \ShEx~\cite{SBLGHPS15} and illustrate it on
the example of a simple graph $G_0$ and a schema $S_0$ in
Figure~\ref{fig:graph-and-schema}. A \emph{typing} of a simple graph $G$ w.r.t.\
$S$ is a relation $T\subseteq N_G\times\Gamma_S$, and we say that $n$ has, or
satisfies, type $t$ (w.r.t.\ $T$) if $(n,t)\in T$. For instance, a typing of
$G_0$ w.r.t.\ $S_0$ is (for clarity, we employ functional notation)
\begin{align*}
  &T_0(n_0) = \{ t_0 \},\ \,
   T_0(n_1) = \{ t_1, t_2\},\ \,
   T_0(n_2) = \{ t_3 \}.
\intertext{The \emph{signature} of a node $n\in N_G$ w.r.t.\ $T$ is 
an $\RBE$ expression}
&\sign_G^T(n)=\bigshuffle_{e\in\out_G(n)} 
\big(\big|_{t\in T(\target_G(e))} \lab_G(e)\dbl t\big) 
\intertext{For instance, the signature of $n_1$ in $G_0$ w.r.t.\ $T_0$ is}
&\sign_{G_0}^{T_0}(n_1) = 
(b\dbl{}t_1\mid b\dbl t_2) \shuffle c\dbl t_3.
\end{align*}
A node $n$ \emph{satisfies} a shape expression $E$ w.r.t.\ a typing $T$ iff
$L(\sign_G^T(n))\cap L(E)\neq\emptyset$. For instance, $n_1$ satisfies the type
$t_2$ of $S_0$ w.r.t.\ $T_0$. The typing $T$ is \emph{valid} if and only if
every node satisfies the type definition of every type assigned to the node
i.e., $L(\sign_G^T(n))\cap L(\delta_S(t))\neq\emptyset$ for every $(n,t)\in
T$. Typings of $G$ w.r.t.\ $S$ form a semi-lattice, with the union as the meet
operation, and consequently there exists a unique maximal typing, which we
denote by $\Typing_{G:S}$. Now, $G$ \emph{satisfies} $S$ if every node of $G$
satisfies at least one type i.e., $\dom(\Typing_{G:S})=N_G$. By $L(S)$ we denote
the set of all simple graphs that satisfy $S$.

\paragraph{Containment}
In this paper, we investigate the containment problem for \ShEx: given two shape
expression schemas $S$ and $S'$ we say that $S$ is \emph{contained} in $S'$, in
symbols $S\subseteq S'$, if $L(S)\subseteq L(S')$. A \emph{counter-example} (for
the containment of $S$ in $S'$) is any graph $G\in L(S)\minus L(S')$.


\section{Embeddings}
\label{sec:embeddings}
In this section we present an alternative definition of the semantics of a
subclass of shape expressions schemas $\ShEx(\RBEZero)$ with a natural notion of
embeddings based on graph simulation. We investigate the relationship between
embeddings and containment, and finally, we study the computational complexity
of embeddings.
\begin{definition}[Embedding]
  \label{def:embedding}
  Given two graphs $G$ and $H$, a binary relation $R\subseteq N_G\times N_H$ is
  a \emph{simulation} of $G$ in $H$ iff for any $(n,m)\in R$ there exists a
  \emph{witness} of simulation of $n$ by $m$ w.r.t.\ $R$ i.e., a function
  $\lambda:\out_G(n)\rightarrow\out_H(m)$ such that for every $e\in\out_G(n)$
  \begin{enumerate}
  \item[1.] $\lab_G(e) = \lab_H(\lambda(e))$,
  \item[2.] $(\target_G(e),\target_H(\lambda(e)))\in R$,
  \end{enumerate}
  and for every $f\in\out_H(m)$ 
  \begin{enumerate}
  \item[3.] $\bigoplus\{\occur_G(e)\mid e\in E_G,\ \lambda(e)=f\}\subseteq\occur_H(f)$. 
  \end{enumerate}
  An \emph{embedding} of $G$ in $H$ is a simulation $R$ of $G$ in $H$ that
  $\dom(R)=N_G$, and we write $G\preccurlyeq H$ if $G$ can be embedded in
  $H$. \qed
\end{definition}
Figure~\ref{fig:embedding} presents an example of an embedding of the graph
$G_0$ and the shape graph $H_0$ that corresponds to the schema $S_0$ in
Figure~\ref{fig:graph-and-schema}.
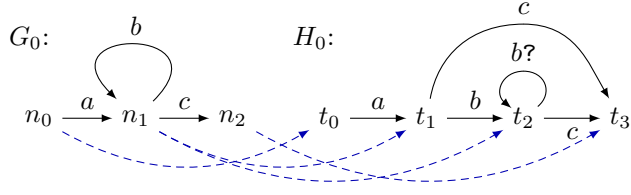
\begin{figure}[htb]
  \centering
  \begin{tikzpicture}[>=latex, scale=0.85,
    punkt/.style={circle,minimum size=0.1cm,draw,fill,inner sep=0pt, outer sep=0.125cm}
    ]
    \node at (-0.15, 1.25) {$G_0$:};
    \node[] (n0) at (0,0) {$n_0$};
    \node[] (n1) at (1.5,0) {$n_1$};
    \node[] (n2) at (3,0) {$n_2$};
    \draw (n0) edge[->] node[above] {$a$} (n1);
    \draw (n1) edge[->] node[above] {$c$} (n2);
    \draw (n1) edge[loop,->] node[above] {$b$} (n1);

    \node at (4.25, 1.25) {$H_0$:};
    \begin{scope}[xshift=4.5cm]
    \node[] (m0) at (0,0) {$t_0$};
    \node[] (m1) at (1.5,0) {$t_1$};
    \node[] (m2) at (3,0) {$t_2$};
    \node[] (m3) at (4.5,0) {$t_3$};
    \path (m2) node[above=0.5cm] (t2p) {};
    \path (t2p) node[above=.05ex] {$b\MAYBE$};
    \path (t2p) node[above=.5cm] (t3p) {};
    \path (t3p) node[above=.05ex] {$c$};

    \draw (m0) edge[->] node[above] {$a$} (m1);
    \draw (m1) edge[->] node[above] {$b$} (m2);
    \draw (m2) edge[-, out=40, in=0, shorten <=-3pt, shorten >=-3.6pt] (t2p);
    \draw (t2p) edge[->, out=180, in=140, shorten <=-3.6pt, shorten >=-4pt] (m2);

    \draw (m1) edge[-, out=80, in=180,  shorten >=-3.6pt] (t3p);
    \draw (t3p) edge[->, out=0, in=120, shorten <=-3.6pt] (m3);
    \draw (m2) edge[->] node[below] {$c$} (m3);
    \end{scope}

    \draw[blue!70!black, densely dashed] (n0) edge[->, bend right] (m0); 
    \draw[blue!70!black, densely dashed] (n1) edge[->, bend right] (m1); 
    \draw[blue!70!black, densely dashed] (n1) edge[->, bend right] (m2); 
    \draw[blue!70!black, densely dashed] (n2) edge[->, bend right] (m3); 

  \end{tikzpicture}
  \caption{\label{fig:embedding} An embedding of $G_0$ in $H_0$.}
  \label{fig:graphs}
\end{figure}

The set of simulations of $G$ in $H$ is a semi-lattice (with the meet operation
interpreted with the set union), and consequently, there exists exactly one
\emph{maximal simulation} of $G$ in $H$. We use embeddings to treat graphs as
schemas. The \emph{language} of a graph $H$ is the set of all simple graphs that
can be embedded in $H$ i.e., $L(H) = \{G\in \GZero \mid G\preccurlyeq H\}$. Two
graphs $G$ and $H$ are \emph{equivalent}, in symbols $G\equiv H$, if
$L(G)=L(H)$. We write $G\subseteq H$ if there is the corresponding language
inclusion $L(G)\subseteq L(H)$. 

There is a natural correspondence between shape expression schemas using
\RBEZero only and shape graphs. The existence of a witness of a simulation is
equivalent to type satisfaction, which shows the following.
\begin{proposition}
\label{prop:shex-zero}
\ShExZero captures precisely $\ShEx(\RBEZero)$.
\end{proposition}
Using arbitrary intervals in graphs does not make them more expressive than
shape graphs which use basic intervals only. Indeed, in the context of shape
expressions any interval can be expressed with basic intervals e.g.,
$a\dbl{}t^{[2;3]}$ is equivalent to
$a\dbl{}t^\ONE\shuffle a\dbl{}t^\ONE \shuffle a\dbl{}t^\MAYBE$. As we show later
on, the restriction to basic intervals has beneficial computational
consequences, which merely reflects the fact that embedding is a sufficient but
not a necessary condition for containment.


\subsection{Relationship to containment}
It is relatively easy to see that embeddings can be composed and consequently we
get.
\begin{lemma}
  \label{lemma:embedding-implies-inclusion}
  For any two graphs $G$ and $H$, $G\preccurlyeq H$ implies
  $G\subseteq H$.
\end{lemma}
The converse does not hold as illustrated in
Figure~\ref{fig:embedding-does-not-imply-inclusion}, where two equivalent graphs
are given but embedding holds only in one direction.
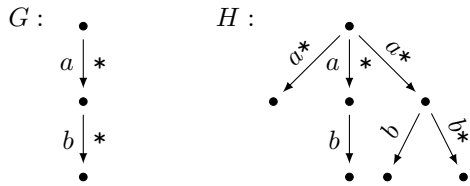
\begin{figure}[htb]
\centering
\begin{tikzpicture}[
  >=latex,
  scale=1,
  punkt/.style={circle,minimum size=0.1cm,draw,fill,inner sep=0pt, outer sep=0.125cm}]
  \begin{scope}
    \node at (-0.75,.125) {$G:$};
    \node[punkt] (t0) at (0,0) {}; 
    \node[punkt] (t1) at (0,-1) {} 
         edge[<-]
         node[left] {$a$}
         node[right] {\MANY} (t0);
    \node[punkt] (t2) at (0,-2) {} 
         edge[<-]
         node[left] {$b$}
         node[right] {\MANY}
         (t1);
  \end{scope}
  \begin{scope}[xshift=4cm]
    \node at (-2,.125) {$H:$};
    \node[punkt] (t0) at (-0.5,0) {}; 
    \begin{scope}[xshift=-1.5cm]
    \node[punkt] (t1) at (0,-1) {} 
         edge[<-]
         node[above, sloped] {$a\MANY$} (t0);
    \end{scope}

    \begin{scope}[xshift=-0.5cm]
    \node[punkt] (t1) at (0,-1) {} 
         edge[<-]
         node[left] {$a$}
         node[right] {\MANY} (t0);
    \node[punkt] (t2) at (0,-2) {} 
         edge[<-]
         node[left] {$b$} (t1);
    \end{scope}

    \begin{scope}[xshift=0.5cm]
    \node[punkt] (t1) at (0,-1) {} 
         edge[<-]
         node[above, sloped] {$a\MANY$} (t0);
    \node[punkt] (t2) at (-0.5,-2) {} 
         edge[<-]
         node[above, sloped] {$b$} (t1);
    \node[punkt] (t2') at (0.5,-2) {} 
         edge[<-]
         node[above, sloped] {$b\MANY$}
         (t1);
    \end{scope}
\end{scope}
\end{tikzpicture}
\caption{Inclusion does not imply an embedding: 
  $G\subseteq H$ but $G \not\preccurlyeq H$.}
  \label{fig:embedding-does-not-imply-inclusion}
\end{figure}
This example basically illustrates that a shape expression $b\dbl{}t^\MANY$ is
equivalent to $\epsilon \mid b\dbl{}t \mid b\dbl{}t^\PLUS$, a (disjoint) union
that enumerates cases of the original expression. It is possible to identify
simple classes of shape graphs, such as schemas that use the occurrence interval
$\MANY$ only, for which embedding is also a necessary condition for containment
but the practical applications of such a class seem to be limited. Later on
(Section~\ref{sec:determinism}) we identify a more practical subclass of shape
graphs for which containment can be decided with the help of embeddings. First
we investigate the computational cost of constructing embeddings.

\subsection{Complexity}
While general graphs with arbitrary intervals do not surpass the expressive
power of graphs using basic intervals only, the computational implications of
using arbitrary intervals are significant. Testing embeddings for shape graphs,
which use only basic occurrence intervals is tractable and it becomes
intractable if arbitrary intervals may be used. This rise in computational
complexity does not come from binary encoding of intervals, in fact the results
remain negative even if the arbitrary intervals are encoded in unary.


The basic algorithm for constructing an embedding of $G$ in $H$ follows a
natural fix-point refinement scheme:
\begin{enumerate}
\itemsep 0pt
\item Begins with $R_0=N_G\times N_H$, 
\item Iteratively refine it $R_i=\Refine(R_{i-1})$ by removing any pair of
  nodes with no simulation  witness
  \[
    \begin{split}
      &\Refine(R)=\{ (n,m) \in R  \mid\text{there exists a} \\
      &\hspace{2.75em}\text{witness $\lambda$
        of simulation of $n$ by $m$ w.r.t.\ $R$}\},
    \end{split}
  \] 
\item Terminates at the earliest iteration $k$ when a fix-point is reached
  $\Refine^*(R_0)=R_k=\Refine(R_k)$.
\end{enumerate}
It is easy to show that the fix-point is the maximal simulation of $G$ in $H$
and naturally, it is an embedding if its domain contains all nodes of $G$. This
process is monotone, removing a number of pairs at each step, and therefore, it
terminates after a polynomial number of steps. The core difficulty, and a
potential source of intractability, is in testing the existence of a witness of
simulation.

\subsubsection{Tractability for basic intervals}
Finding witnesses of simulation is tractable when basic intervals are used only.
\begin{theorem}
  \label{thm:embeddings-tractable}
  Testing the existence of embeddings between shape graphs is in \PTIME.
\end{theorem}
\begin{proof}
We fix two graphs $G$ and $H$, a relation $R\subseteq N_G\times N_H$, and a pair
of nodes $(n,m)\in R$. We abstract the problem of existence of a witness of
simulation of $n$ by $m$ w.r.t.\ $R$ as a \emph{flow routing} problem, where we
are given a set of sources $V=\out_G(n)$, a set of sinks $U=\out_H(m)$, and a
source-to-sink connection table
$E=\{(v,u)\in V\times U \mid \lab_G(v)=\lab_H(u) \land
(\target_G(v),\target_H(u))\in R\}$. Additionally, every source $v\in V$ outputs
a volume of water
between $v\dotmin = \min(\occur_G(v))$ and $v\dotmax = \max(\occur_G(v))$, and
every sink $u\in U$ requires an input of
at least $u\dotmin = \min(\occur_H(u))$ and not more than
$u\dotmax = \max(\occur_H(u))$. The flow routing problem is to find a valid
\emph{routing} $\lambda:V\rightarrow U$ i.e., a routing such that
$(v,f(v))\in E$ for every source $v\in V$ and there are no deficits or overflows
at any source. More formally, given a routing $\lambda$ we estimate the inflow
at a sink $u$ with
\begin{align*}
&\mininflow_\lambda(u)=\textstyle\sum_{\lambda(v)=u} v\dotmin,\\
&\maxinflow_\lambda(u)=\textstyle\sum_{\lambda(v)=u} v\dotmax. 
\end{align*}
A sink $u$ is in \emph{deficit} if $\mininflow_\lambda(u)< u\dotmin$ and $u$ is
in \emph{overflow} if $\maxinflow_\lambda(u) > u\dotmax$. Observe that the
conditions 1 and 2 in Definition~\ref{def:embedding} are ensured by the
definition of $E$ while the condition 3 follows from lack of deficits and
overflows. Also, w.l.o.g.\ we can assume that in $E$ every source is paired with
at least one sink.

The algorithm for constructing a valid routing works as follows: 
\begin{enumerate}
\itemsep 0pt
\item it starts with an empty routing,
\item it assigns a sink to every source while distributing any overflow
  by pushing it forth to other sinks,
\item it solves any deficit at a sink by pulling back the input from sources
  assigned to other sinks.
\end{enumerate}
The main reason why this approach is successful if the use of basic occurrence
intervals in shape graphs, which implies that the lower bounds are only $0$ and
$1$ while the upper bounds are $1$ and $\infty$. When constructing the routing
$\lambda$ we need to pay attention to \emph{saturated} sinks that are unable to
accept any additional inflow. However, saturated sinks are exactly those $u$'s
with $u\dotmax=1$ and $\maxinflow_\lambda(u)=1$. Furthermore, w.l.o.g.\ we can
assume that $v\dotmax \leq u\dotmax$ for $(v,u)\in E$, and in particular a
source with $\infty$ upper bound can only be routed to a sink with upper bound
$\infty$. Consequently, any overflow created by the algorithm at a sink $u$ is
\emph{singular} i.e., $\maxinflow_\lambda(u)=2$ and $u\dotmax=1$

Given a (partial) routing $\lambda$ and a source $v$ with no assigned sink, the
algorithm assigns to $v$ any admissible sink $u_0$ i.e., such that
$(v,u_0)\in E$. If an overflow is created at $u_0$, the algorithm attempts to
find an acyclic path $\pi$ from $u_0$ to $\fin$ in the \emph{push-forth graph}
$G_\lambda^\rightarrow=(N,A)$, where the nodes are $N=V\cup U \cup \{\fin \}$
and the oriented edges $A$ are (for $v\in V$ and $u\in U$):
\begin{itemize}
\item $u\rightarrow v$ if $\lambda(v) = u$ and $u$ is saturated; an additional
  inflow of $1$ at sink $u$ must be redirected further and this can be done by
  redirecting the output of $v$ to another sink.
\item $v\rightarrow u$ if $(v,u)\in E$ but $\lambda(v)\neq u$; the source $v$
  can be routed to $u$ and any additional inflow at $u$ is at most $1$.
\item $u\rightarrow \fin$ if $u$ is not saturated; the sink can accept an
  additional inflow of $1$.
\end{itemize}
Rerouting $\lambda$ in accordance with a path from $u$ to $\fin$ gives us a
overflow-free routing.

When a total overflow-free routing $\lambda$ is constructed, the algorithm
identifies any sink $u_0$ with a deficit and tries to solve it by finding an
acyclic path $\pi$ from $u_0$ to $\fin$ in the \emph{pull-back graph}
$G_\lambda^\leftarrow=(N,A)$, where the nodes are $N=V\cup U\cup\{\fin\}$, and
oriented edges $A$ are (for $v\in V$ and $u\in U$):
\begin{itemize}
\item $u\rightarrow v$ if $\lambda(v)\neq u$ and $v\dotmin=1$; rerouting $v$ to
  $u$ will solve a deficit of $1$ at $u$ and may create a overflow at $u$ but
  only if $u=u_0$ and then the overflow is singular.
\item $v\rightarrow u$ if $\lambda(v)=u$, $u\dotmin=1$, and $v$ is the only
  source such that $\lambda(v)=u$ and $v\dotmin=1$; rerouting $v$ away from $u$
  will create a deficit of $1$ at $u$.
\item $v\rightarrow \fin$ if $\lambda(v)\dotmin\neq 1$, $v\dotmin=1$, and there
  is $v'\neq v$ such that $\lambda(v')=\lambda(v)$ and $v'\dotmin=1$; rerouting
  the source $v$ from the sink $\lambda(v)$ will not create a deficit at
  $\lambda(v)$.
\end{itemize}
Rerouting $\lambda$ in accordance with $\pi$ renders $\lambda$ deficit-free at
$u_0$. If the rerouting creates a singular overflow at $u_0$, the algorithm uses
the push-forth graph $G_\lambda^\rightarrow$ to find an acyclic path $\pi'$ from
$u_0$ to $\fin$ that is deficit-free i.e., with no edge $u\rightarrow v$ such
that $u\dotmin=v\dotmin=1$, which guarantees that further rerouting $\lambda$ in
accordance with $\pi'$ yields a overflow-free routing with one sink node $u_0$
less in deficit.

Naturally, the algorithm is polynomial because the sizes of the push-forth and
pull-back graphs are bounded by the size of $E$, and all constructed paths are
acyclic. \qed
\end{proof}


\subsubsection{Intractability for arbitrary intervals}
The flow routing problem becomes intractable if arbitrary intervals are allowed.

\begin{theorem}
  \label{thm:embeddings-intractable}
  Testing the existence of embeddings between graphs with arbitrary intervals is
  \NP-complete.
\end{theorem}
\begin{proof}
The problem is in NP because it suffices to
guess the embedding and all witnesses and verify that they are valid (which can
be done in polynomial time). We prove hardness with a reduction from SAT. We
take a CNF formula $\varphi=c_1\land\ldots\land c_m$ over variables
$x_1,\ldots,x_n$, and w.l.o.g.\ we can assume that each variable occurs exactly
$k$ times and there is at least one positive and one negative occurrence of each
variable. We construct two graphs with arbitrary intervals but for clarity of
presentation we employ the \ShEx syntax. The first graph $H$ is defined with the
rules (with $i\in\{1,\ldots,n\}$ and $j\in\{1,\ldots,k\}$)

\begin{align*}
  & \begin{aligned}
    \spec{r}_1 \rightarrow {} 
    & \spec{a}\dbl{\spec{w}_1^{[k; k]}} \shuffle 
    \spec{a}\dbl{\spec{x}_{1,1}} \shuffle \ldots \shuffle \spec{a}\dbl {\spec{x}_{1,k}} \shuffle \\
    & \spec{a}\dbl{\neg \spec{x}_{1,1}} \shuffle \ldots \shuffle \spec{a}\dbl {\neg \spec{x}_{1,k}}\shuffle\\
    & \ldots\\
    & \spec{a}\dbl{\spec{w}_n^{[k; k]}} \shuffle 
    \spec{a}\dbl{\spec{x}_{n,1}} \shuffle \ldots \shuffle \spec{a}\dbl {\spec{x}_{n,k}} \shuffle\\
    & \spec{a}\dbl{\neg \spec{x}_{n,1}} \shuffle \ldots \shuffle \spec{a}\dbl {\neg \spec{x}_{n,k}}
  \end{aligned}\\
  &\begin{aligned}
  & {\spec{x}_{i,j}} \rightarrow {\spec{x}_{i,j}}\dbl\spec{o}&\qquad&
    {\neg \spec{x}_{i,j}} \rightarrow {\neg \spec{x}_{i,j}}\dbl\spec{o}\\
  &\spec{w}_i \rightarrow \spec{v}_i\dbl\spec{o} &&
 \spec{o} \rightarrow \epsilon
  \end{aligned}
\intertext{The second graph $K$ is defined with the rules (with $i\in\{1,\ldots,n\}$)}
  &\begin{aligned}
    \spec{r}_2 \rightarrow {}&
    \spec{a}\dbl{\spec{x}_1}^{[k, k]} \shuffle \spec{a}\dbl{\neg \spec{x}_1}^{[k, k]} \shuffle \\
    & \ldots \\
    & \spec{a}\dbl{\spec{x}_n}^{[k, k]}\shuffle \spec{a}\dbl{\neg \spec{x}_n}^{[k, k]} \shuffle \\
    & \spec{a}\dbl {\spec{c}_1}^\PLUS \shuffle  \ldots \shuffle \spec{a}\dbl {\spec{c}_s}^\PLUS
  \end{aligned}\\
  & \spec{x}_i \rightarrow \spec{v}_i\dbl\spec{o}^\MAYBE \shuffle
    {\spec{x}_{i,1}}\dbl\spec{o}^\MAYBE \shuffle \ldots \shuffle {\spec{x}_{i,k}}\dbl\spec{o}^\MAYBE \\
  &{\neg \spec{x}_i} \rightarrow \spec{v}_i\dbl\spec{o}^\MAYBE \shuffle {\neg
    \spec{x}_{i,1}}\dbl\spec{o}^\MAYBE \shuffle \ldots \shuffle
    {\neg \spec{x}_{i,k}}\dbl\spec{o}^\MAYBE\\
  &\spec{o} \rightarrow \epsilon\\
  \intertext{and additionally for $p\in\{1,\ldots,m\}$ the clause
  $c_p=\ell_1\lor\ldots\lor\ell_s$ yields the rule} 
  & \spec{c}_p \rightarrow
    \spec{l}_1\dbl\spec{o}^\MAYBE
    \shuffle
    \spec{l}_2\dbl\spec{o}^\MAYBE
    \shuffle \ldots \shuffle
    \spec{l}_s\dbl\spec{o}^\MAYBE
\end{align*}   
where $\spec{l}_q$ is $\spec{x}_{i,j}$ if $\ell_q=x_i$ or $\lnot \spec{x}_{i,j}$
if $\ell_q=\lnot x_i$, and here, $\ell_q$ is the $j$-th occurrence of $x_i$ in
$\varphi$.

Clearly, for $i\in\{1,\ldots,n\}$ and $j\in\{1,\ldots,k\}$ the node
$\spec{x}_{i,j}$ of $H$ is simulated by the node $\spec{x}_i$ of $K$,
$\lnot\spec{x}_{i,j}$ by $\lnot\spec{x}_i$, $\spec{w}_i$ by both $\spec{x}_i$
and $\lnot\spec{x}_i$, and $\spec{o}$ by $\spec{o}$. Furthermore,
$\spec{x}_{i,j}$ ($\lnot \spec{x}_{i,j}$) is simulated by any $\spec{c}_p$ that
corresponds to a clause $c_p$ that uses $x_i$ ($\lnot x_i$ resp.). We claim that
$\varphi$ is satisfiable iff $\spec{r}_1$ is simulated by $\spec{r}_2$, and
therefore, iff $H\preccurlyeq K$.

For a satisfying valuation $v$ of $\varphi$, we define a witness
$\lambda: \out_H(\spec{r}_1) \to \out_K(\spec{r}_2)$ by specifying the mapping
between target nodes of the edges (since no target node is used more than once
this is unambiguous). For $i\in\{1,\ldots,n\}$, if $v(x_i)=1$, then we map
$\spec{w}_i$ to $\spec{x}_i$, and for $j\in\{1,\ldots,k\}$ we map
$\neg \spec{x}_{i,j}$ to $\neg \spec{x}_i$ and $\spec{x}_{i,j}$ to $\spec{c}_p$
that corresponds to the clause $c_p$ with the $j$-th occurrence of $x_i$ if this
occurrence is positive, or otherwise to any clause $c_p$ with a positive
occurrence of $x_i$; if $v(x_i)=0$, then we map $\spec{w}_i$ to $\spec{x}_i$,
and for $j\in\{1,\ldots,k\}$ we map $\spec{x}_{i,j}$ to $\neg \spec{x}_i$ and
$\neg \spec{x}_{i,j}$ to $\spec{c}_p$ that corresponds to the clause $c_p$ with
the $j$-th occurrence of $x_i$ if this occurrence it is negative, or to any
clause $c_p$ with a negative occurrence of $x_i$; It is easy to verify that if
$v$ satisfies $\varphi$, $\lambda$ is a simulation witness of $\spec{r}_1$ by
$\spec{r}_2$.

Conversely, we take a witness
$\lambda: \out_H(\spec{r}_1) \to \out_K(\spec{r}_2)$ and construct a satisfying
valuation $v$ by inspecting the implied mapping on the target nodes of
$\spec{r}_1$ and $\spec{r}_2$. We observe that for $i\in\{1,\ldots,k\}$ because
of the occurrence interval $[k; k]$ the node $\spec{w}_i$ must be mapped to
either the node $\spec{x}_i$ or $\lnot\spec{x}_i$, and we set $v(x_i)=1$ or
$v(x_i)=0$ respectively. W.l.o.g. assume $\spec{w}_i$ is mapped to
$\spec{x}_i$. Because of the occurrence interval $[k; k]$, no other node can be
mapped to $\spec{x}_i$, which means that all nodes $\spec{x}_{i,j}$ must be
mapped to nodes $\spec{c}_p$. Because of the occurrence interval $\PLUS$, at
least one $\spec{x}_{i,j}$ or $\lnot \spec{x}_{i,j}$ must be mapped to every
$\spec{c}_p$, and consequently, each clause of $\varphi$ is satisfied by the
valuation $v$. \qed
\end{proof}


\section{Determinism}
\label{sec:determinism}
In this section we identify a subclass of deterministic shape expression schemas
that is arguably of practical use and for which embedding is equivalent to
containment, and therefore, enjoy desirable computational properties. We recall
that deterministic shape expression schemas (\DetShEx) forbid using the same
edge label with more than one type in a single type
definition~\cite{SBLGHPS15}. While for deterministic shape graphs using only
$\ONE$ and $\MANY$ embeddings capture containment, adding $\MAYBE$ requires
careful handling or the containment becomes intractable as we show later on. The
proposed approach can further be extended to allow $\PLUS$, however, the
additional restrictions this introduces render the obtained class impractical,
and consequently, we drop $\PLUS$ altogether. In many practical situations,
using $\MANY$ instead of $\PLUS$ is acceptable. 

Now, given a shape graph $H$ and a type $t\in N_H$, a \emph{reference} to $t$ is
any edge $e\in E_H$ that leads to $t$ i.e. $\target_H(e)=t$. A reference $e$ is
\emph{$\MANY$-closed} if $\occur_H(e)=\MANY$ or all references to $\source_H(e)$
are $\MANY$-closed.
\begin{definition}
  A shape graph $H$ is \emph{deterministic} if for every node $n\in N_H$ and
  every label $a\in\Sigma$, $n$ has at most one outgoing edge labeled with
  $a$. By $\DetShExZero$ we denote the class of all deterministic shape
  graphs. By $\DetShExZeroMinus$ we denote the class of deterministic shape
  graphs that do not use $\PLUS$ and any type using $\MAYBE$ is referenced at
  least once and all references to it are $\MANY$-closed. \qed
\end{definition}
Intuitively, we require that any type using $\MAYBE$ must be referenced and can
only be referenced (directly or indirectly) through $\MANY$. The schema in
Figure~\ref{fig:rdf-bug-reports} belongs to \DetShExZeroMinus since both uses of
the \MAYBE operator are closed by the edge $\mathtt{related}$ with interval
$\MANY$.

Interestingly, the class $\DetShExZeroMinus$ allows construction of graphs that
characterize any schema in $\DetShExZeroMinus$ up to containment.
\begin{lemma}
  \label{lemma:counter-example-det-shex-zero-down}
  For any $H\in\DetShExZeroMinus$, there exists a simple graph $G\in L(H)$ of size polynomial in the size of $H$ such that for any $K\in\DetShExZeroMinus$ we have that $G\preccurlyeq K$ implies $H\preccurlyeq K$.
\end{lemma}
The precise construction of the graph $G$ that characterizes $H$ is in appendix
and here we outline the main ideas and illustrate them on an example in
Figure~\ref{fig:detshex-characteristic-graph}.
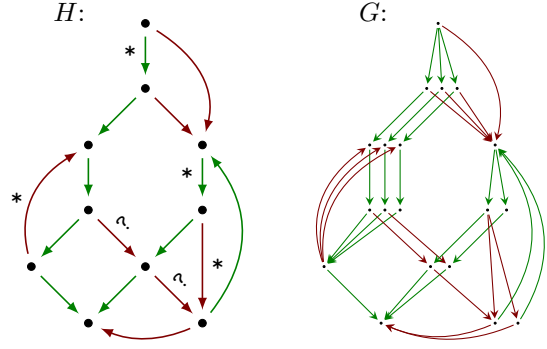
\begin{figure}[hbt]
  \centering
  \begin{tikzpicture}[xscale=1,yscale=0.75]
    \node at (-1,0.5) {$H$:};
    \begin{scope}[
    >=latex,
    punkt/.style={circle,minimum size=0.1cm,draw,fill,
      inner sep=0pt, outer sep=0.125cm},
    a/.style={color=green!50!black,semithick},
    b/.style={color=red!50!black,semithick}
      ]
    \small
    \node[punkt] (n0) at (0,0.3) {};
    \node[punkt] (n1) at (0,-0.85) {};
    \node[punkt] (n2) at (-0.75,-1.85) {};
    \node[punkt] (n3) at (0.75,-1.85) {};
    \node[punkt] (n4) at (-0.75,-3) {};
    \node[punkt] (n5) at (0.75,-3) {};
    \node[punkt] (n6) at (-1.5,-4) {};
    \node[punkt] (n7) at (0,-4) {};
    \node[punkt] (n8) at (-0.75,-5) {};
    \node[punkt] (n9) at (0.75,-5) {};
    
    \draw[a] (n0) edge[a,->] node[pos=0.4,left,black] {\MANY} (n1);
    \draw[b] (n0) edge[->,bend left] (n3);
    \draw[a] (n1) edge[->] (n2);
    \draw[b] (n1) edge[->] (n3);
    \draw[a] (n2) edge[->] (n4);
    \draw[a] (n3) edge[->] node[pos=0.4,left,black] {\MANY} (n5);
    \draw[a] (n4) edge[->] (n6);
    \draw[b] (n4) edge[->] node[pos=0.4,right,black,rotate=45] {\MAYBE} (n7);
    \draw[a] (n5) edge[->] (n7);
    \draw[b] (n5) edge[->] node[pos=0.48,right,black] {\MANY} (n9);
    \draw[b] (n6) edge[->,bend left] node[pos=0.49,left,black] {\MANY} (n2);
    \draw[a] (n6) edge[->] (n8);
    \draw[a] (n7) edge[->] (n8);
    \draw[b] (n7) edge[->] node[pos=0.4,right,black,rotate=45] {\MAYBE} (n9);
    \draw[a] (n9) edge[->,bend right] (n3);
    \draw[b] (n9) edge[->,bend left] (n8);
    \end{scope}
    \node at (3,0.5) {$G$:};
    \begin{scope}[xshift=4cm,
      >=stealth,
      punkt/.style={circle,minimum size=0.025cm,draw,fill,
        inner sep=0pt, outer sep=0.05cm},
      a/.style={color=green!50!black},
      b/.style={color=red!50!black}
      ]


    \node[punkt] (n00) at (-0.15,0.3) {};
    \node[punkt] (n10) at (-0.3,-0.85) {};
    \node[punkt] (n11) at (-0.1,-0.85) {};
    \node[punkt] (n12) at (+0.1,-0.85) {};
    \node[punkt] (n20) at (-1.05,-1.85) {};
    \node[punkt] (n21) at (-0.85,-1.85) {};
    \node[punkt] (n22) at (-0.65,-1.85) {};
    \node[punkt] (n30) at (+0.6,-1.85) {};
    \node[punkt] (n40) at (-1.05,-3) {};
    \node[punkt] (n41) at (-0.85,-3) {};
    \node[punkt] (n42) at (-0.65,-3) {};
    \node[punkt] (n50) at (+0.5,-3) {};
    \node[punkt] (n51) at (+0.75,-3) {};
    \node[punkt] (n60) at (-1.65,-4) {};
    \node[punkt] (n70) at (-0.25,-4) {};
    \node[punkt] (n71) at (+0.00,-4) {};
    \node[punkt] (n80) at (-0.90,-5) {};
    \node[punkt] (n90) at (+0.60,-5) {};
    \node[punkt] (n91) at (+0.90,-5) {};
    
    \draw[a] (n00) edge[->] (n10);
    \draw[a] (n00) edge[->] (n11);
    \draw[a] (n00) edge[->] (n12);

    \draw[b] (n00) edge[->,bend left] (n30);

    \draw[a] (n10) edge[->] (n20);
    \draw[a] (n11) edge[->] (n21);
    \draw[a] (n12) edge[->] (n22);

    \draw[b] (n10) edge[->] (n30);
    \draw[b] (n11) edge[->] (n30);
    \draw[b] (n12) edge[->] (n30);

    \draw[a] (n20) edge[->] (n40);
    \draw[a] (n21) edge[->] (n41);
    \draw[a] (n22) edge[->] (n42);

    \draw[a] (n30) edge[->] (n50);
    \draw[a] (n30) edge[->] (n51);

    \draw[a] (n40) edge[->] (n60);
    \draw[a] (n41) edge[->] (n60);
    \draw[a] (n42) edge[->] (n60);

    \draw[b] (n40) edge[->] (n70);
    \draw[b] (n41) edge[->] (n71);

    \draw[a] (n50) edge[->] (n70);
    \draw[a] (n51) edge[->] (n71);

    \draw[b] (n50) edge[->] (n90);
    \draw[b] (n50) edge[->] (n91);

    \draw[b] (n60) edge[->,bend left] (n20);
    \draw[b] (n60) edge[->,bend left] (n21);
    \draw[b] (n60) edge[->,bend left] (n22);

    \draw[a] (n60) edge[->] (n80);

    \draw[a] (n70) edge[->] (n80);
    \draw[a] (n71) edge[->] (n80);

    \draw[b] (n70) edge[->] (n90);

    \draw[a] (n90) edge[->,bend right] (n30);
    \draw[a] (n91) edge[->,bend right] (n30);

    \draw[b] (n90) edge[->,bend left] (n80);
    \draw[b] (n91) edge[->,bend left] (n80);
    \end{scope}
  \end{tikzpicture}
  \caption{Characterizing example. Different colors denote different labels.}
  \label{fig:detshex-characteristic-graph}
\end{figure}
 
%
In essence, for every type $t\in N_H$ the graph $G$ needs to contain a number of
nodes of type $t$ that serve the purpose of characterizing $t$. If $t$ has an
outgoing \MANY-edge $e$ labeled with $a$ that leads to the type $s$, then at
least one node $n$ of $G$ that characterizes $t$ needs to have at least two
outgoing edges labeled with $a$ that lead to nodes $A$ that characterize the
type $s$. When $n$ is mapped to a type $t'$ of $K$ that has an outgoing edge
$e'$ labeled with $a$ and leading to $s'$, all $a$-children $A$ must to be
mapped to $s'$. This shows that $e'$ is an \MANY-edge. Interestingly, this
observation also applies to descendants of $A$.  If the type $s$ has an outgoing
edge with label $b$ that leads to the type $u$, then any $b$-child of a node in
$A$ must have the type $u$, and furthermore, they are all mapped to a type $u'$
that is reachable from $s'$ with an edge labeled with $b$, etc. 

Now, for a type $t$ with an outgoing $\MAYBE$-edge labeled with $a$ we need two
nodes in $G$ that characterize $t$, which guarantee that the corresponding type
in $K$ uses the right occurrence interval: one node with one outgoing edge with
label $a$ and one node with no such outgoing edge. Naturally, we need to make
sure that those two nodes are mapped to the same type in $K$ and this is
accomplished by making sure there is an ascending path from every \MAYBE-edge to
every closest \MANY-edge. In general, every type in $H$ is characterized by a
number of nodes that is at most $2$ plus the number of $\MAYBE$-edges in
$H$. Lemma~\ref{lemma:counter-example-det-shex-zero-down} renders containment
and embeddings equivalent.
\begin{corollary}
  For $H,K\in\DetShExZeroMinus$, $H\subseteq K$ iff $H\preccurlyeq K$.
\end{corollary}
Consequently, by Theorem~\ref{thm:embeddings-tractable}, we obtain. 
\begin{corollary}
  \label{cor:containment-detshex-tractable}
  Containment for $\DetShExZeroMinus$ is in \PTIME.
\end{corollary}
\begin{figure*}[]
  \tikzset{poz/.style={coordinate}}
  \tikzset{aux/.style={gray, line width=0.25pt, densely dotted}}
  \tikzset{pt/.style={gray, circle, inner sep=0.05cm,fill=gray}}
  \newcommand{\poz}[2]{
    \path #2 node[poz] (#1) {};
  }
  \centering
  \footnotesize
  \begin{tikzpicture}[xscale=1.75,yscale=1.5, >=latex]
    \begin{scope}
      \path[use as bounding box] (0.5,-2.5) rectangle (7,0.5);
      \node at (-1,0) {$H$:};
      \node (r) at (0,0) {\spec{r}};
      \node (v) at (0,-1) {\spec{v}};
      \node (o) at (0,-2) {\spec{o}};
      \draw[bend angle=50] (r) edge[->, bend left] node[left]{$\spec{x}_1$} (v);
      \draw[bend angle=05] (r) edge[->,bend right] node[left]{$\spec{x}_2$} (v);
      \draw[bend angle=65] (r) edge[->, bend right] node[left]{$\spec{x}_3$} (v);
      \draw[bend angle=40] (v) edge[->, bend left] node[left]{$\spec{t}\MAYBE$}(o);
      \draw[bend angle=40] (v) edge[->, bend right] node[left]{$\spec{f}\MAYBE$}(o);
    \end{scope}

    \node at (1,0) {$K$:};
    \begin{scope}[xshift=6cm, yshift=-1cm]



      \node (v0) at (2,0) {$\spec{v}^{0}$};

      \node (r01) at (1,-1) {$\spec{r}^{0}_{1}$};
      \node (r02) at (1,0) {$\spec{r}^{0}_{2}$};
      \node (r03) at (1,1) {$\spec{r}^{0}_{3}$};

      \node (v) at (0,0) {$\spec{v}$};

      \draw[bend angle=20] (r01) edge[->,bend right] node[above, sloped] {$\spec{x}_1$} (v0);
      \draw[bend angle=25] (r01) edge[->,bend left] node[above, sloped] {$\spec{x}_2$} (v);
      \draw[bend angle=05] (r01) edge[->,bend right] node[above, sloped] {$\spec{x}_3$} (v);

      \draw[bend angle=25] (r02) edge[->] node[above, sloped] {$\spec{x}_2$} (v0);
      \draw[bend angle=25] (r02) edge[->,bend right] node[above, sloped] {$\spec{x}_1$} (v);
      \draw[bend angle=25] (r02) edge[->,bend left] node[above, sloped] {$\spec{x}_3$} (v);

      \draw[bend angle=20] (r03) edge[->,bend left] node[above, sloped] {$\spec{x}_3$} (v0);
      \draw[bend angle=25] (r03) edge[->,bend right] node[above, sloped] {$\spec{x}_2$} (v);
      \draw[bend angle=05] (r03) edge[->,bend left] node[above, sloped] {$\spec{x}_1$} (v);

      \node (r11) at (-1,-0.75) {$\spec{r}^1_1$};
      \node (r12) at (-1,0) {$\spec{r}^1_2$};
      \node (r13) at (-1,0.75) {$\spec{r}^1_3$};

      \node (v1) at (-1.66,0) {$\spec{v}^1$};

      \draw[bend angle=20] (r11) edge[->,bend left] node[above, sloped] {$\spec{x}_1$} (v1);
      \draw[bend angle=25] (r11) edge[->,bend right] node[above=-1pt, sloped] {$\spec{x}_2$} (v);
      \draw[bend angle=03] (r11) edge[->,bend left] node[above, sloped] {$\spec{x}_3$} (v);

      \draw[bend angle=20] (r12) edge[->] node[above, sloped] {$\spec{x}_2$} (v1);
      \draw[bend angle=20] (r12) edge[->,bend left] node[above, sloped] {$\spec{x}_1$} (v);
      \draw[bend angle=20] (r12) edge[->,bend right] node[above, sloped] {$\spec{x}_3$} (v);

      \draw[bend angle=20] (r13) edge[->,bend right] node[above=-1pt, sloped] {$\spec{x}_3$} (v1);
      \draw[bend angle=35] (r13) edge[->,bend left] node[pos=0.4,above, sloped] {$\spec{x}_2$} (v);
      \draw[bend angle=03] (r13) edge[->,bend left] node[above, sloped] {$\spec{x}_1$} (v);

      \node (o) at (-2.66,0) {$\spec{o}$};

      \poz{c1}{(0,-1.75)}
      \poz{c2}{(-2.33,-1.05)}

      \draw[->] (v) .. controls (c1) and (c2) .. node[pos=0.5,above, sloped] {$\spec{t}\MAYBE$}(o);

      \poz{c3}{(0,1.75)}
      \poz{c4}{(-2.33,1.05)}

      \draw[->] (v) .. controls (c3) and (c4) .. node[pos=0.5, above, sloped] {$\spec{f}\MAYBE$}(o);

      \draw[bend angle=20] (v1) edge[->,bend left] node[above, sloped] {$\spec{t}$} (o);
      \draw[bend angle=20] (v1) edge[->,bend right] node[above, sloped] {$\spec{f}$} (o);

      \node (vd11) at (-4,-1) {$\spec{v}^1_1$};
      \node (vd12) at (-4,-0.66) {$\spec{v}^1_2$};
      \node (vd13) at (-4,-.33) {$\spec{v}^1_3$};

      \node (vd21) at (-4,0.33) {$\spec{v}^2_1$};
      \node (vd22) at (-4,0.66) {$\spec{v}^2_2$};
      \node (vd23) at (-4,1) {$\spec{v}^2_3$};

      \node (rd1) at (-5,-0.5) {$\spec{r}^{d}_1$};
      \node (rd2) at (-5,0.5) {$\spec{r}^{d}_2$};

      \draw[bend angle=25] (rd1) edge[->,bend right] node[above, sloped] {$\spec{x}_{1}$} (vd11);
      \draw[bend angle=05] (rd1) edge[->,bend right] node[above, sloped] {$\spec{x}_{2}$} (vd12);
      \draw[bend angle=15] (rd1) edge[->,bend left] node[above, sloped] {$\spec{x}_{3}$} (vd13);

      \draw[bend angle=15] (rd2) edge[->,bend right] node[above, sloped] {$\spec{x}_{1}$} (vd21);
      \draw[bend angle=05] (rd2) edge[->,bend left] node[above, sloped] {$\spec{x}_{2}$} (vd22);
      \draw[bend angle=25] (rd2) edge[->,bend left] node[above, sloped] {$\spec{x}_{3}$} (vd23);

      \poz{d111}{(-3.66,-1.05)}
      \poz{d112}{(-2.66,-1.15)}

      \draw[->] (vd11) .. controls (d111) and (d112) .. node[pos=0.235, above, sloped] {$\spec{t}$} (o);

      \poz{d121}{(-3.66,-0.75)}
      \poz{d122}{(-2.8,-0.8)}

      \draw[->] (vd12) .. controls (d121) and (d122) .. node[pos=0.25, above, sloped] {$\spec{f}$} (o);

      \poz{d131}{(-3.66,-0.45)}
      \poz{d132}{(-3.1,-0.4)}

      \draw[->] (vd13) .. controls (d131) and (d132) .. node[pos=0.65,above, sloped] {$\spec{t}\MAYBE$}(o);

      \poz{d131}{(-3.66,-0.15)}
      \poz{d132}{(-3.2,-0.1)}

      \draw[->] (vd13) .. controls (d131) and (d132) .. node[pos=0.35,above, sloped] {$\spec{f}\MAYBE$}(o);

      \poz{d211}{(-3.66,0.15)}
      \poz{d212}{(-3.2,0.1)}

      \draw[->] (vd21) .. controls (d211) and (d212) .. node[pos=0.55,above, sloped] {$\spec{f}\MAYBE$}(o);

      \poz{d211}{(-3.66,0.45)}
      \poz{d212}{(-3.1,0.4)}

      \draw[->] (vd21) .. controls (d211) and (d212) .. node[pos=0.4, above, sloped] {$\spec{t}\MAYBE$}(o);

      \poz{d221}{(-3.66,0.75)}
      \poz{d222}{(-2.8,0.8)}

      \draw[->] (vd22) .. controls (d221) and (d222) .. node[pos=0.25, above, sloped] {$\spec{t}$} (o);

      \poz{d231}{(-3.66,1.05)}
      \poz{d232}{(-2.66,1.15)}

      \draw[->] (vd23) .. controls (d231) and (d232) .. node[pos=0.235, above, sloped] {$\spec{f}$} (o);
    \end{scope}
  \end{tikzpicture}  
  \caption{An example of reduction on $\varphi=(x_1\land \lnot x_2) \lor 
    (x_2\land \lnot x_3)$.}
\label{fig:detshex-reduction}
\end{figure*}  


We observe that lifting the additional restrictions leads to intractability.
\begin{theorem}
  \label{thm:containment-detshex-intractable}
  Containment for $\DetShExZero$ is \coNP-hard.
\end{theorem}
\begin{proof}[sketch]
  We prove it by reduction from tautology of DNF formulas and we illustrate the
  reduction on an example of
  $\varphi=(x_1\land \lnot x_2) \lor (x_2\land \lnot x_3)$. We construct two
  deterministic schemas $H$ and $K$ presented in
  Figure~\ref{fig:detshex-reduction}.  
  
  Essentially, $H$ defines valuations of the variables of $\varphi$: a node with
  the root type $\spec{r}$ has outgoing edges labeled with the name of the
  variable leading to a node of type $\spec{v}$ that represents the value of the
  variable $\spec{t}$ or $\spec{f}$. Because $\DetShExZero$ does not allow
  disjunction, nodes of type $\spec{v}$ may also have both outgoing edges
  $\spec{t}$ and $\spec{f}$, or neither of them. These cases are covered in $K$
  by the types $\spec{r}^1_i$'s and $\spec{r}^0_i$'s respectively. The types
  $\spec{r}^{d}_j$'s capture precisely the valuations that satisfy the clauses
  of $\varphi$. Naturally, $H\subseteq K$ iff $\varphi$ is a tautology. \qed
\end{proof}


\section{Shape graphs}
\label{sec:shex-zero}
In this section we consider shape graphs $\ShExZero$, which correspond to the
subclass $\ShEx(\RBEZero)$ of shape expression schemas that use only $\RBEZero$
expression for type definitions. First, we show that the size of a
counter-example is at most exponential and that the bound is tight. Then, we
investigate the complexity of the containment problem and show that it is
\EXP-complete.

\subsection{Counter-example}
In the view of Theorem~\ref{thm:containment-detshex-intractable}, the immediate
question is whether the size of a counter-example of \ShExZero is bounded by a
polynomial. We show, however, that a counter-example might need to have an
exponential number of nodes.
\begin{lemma}
  \label{lemma:shex-zero-small-example}
  For any $n$, there exist two shape graphs $H$ and $K$ such that $H\not\subseteq K$ and the smallest graph $G\in L(H)\minus L(K)$ is of size exponential in $n$.
\end{lemma}
\begin{proof}
In our construction the counter-example i.e., $G\in L(H)\minus L(K)$, is
essentially a binary tree of depth $n$ modeled with the rules 
(for $i\in\{1,\ldots,n\}$)
\begin{align*}
  &\smash{\spec{t}^{(i)} \rightarrow 
    \spec{L}\dbl{} \spec{t}^{(i+1)} \shuffle 
    \spec{R}\dbl\spec{t}^{(i+1)}}\\[-2pt]
\intertext{
The leaves of this tree store each a subset of
$A=\{\spec{a}_1,\ldots,\spec{a}_n\}$, modeled with the two rules}
&\begin{aligned}
  &\smash{\spec{t}^{(n+1)} \rightarrow 
    \spec{a}_1\dbl{} \spec{t}_\spec{o}^\MAYBE \shuffle 
    \ldots \shuffle
    \spec{a}_n\dbl{} \spec{t}_\spec{o}^\MAYBE} &
  &\spec{t}_\spec{o} \rightarrow \epsilon\\[-2pt]
\end{aligned}
\intertext{
The schema $H$ consists exactly of all the above rules while the schema $K$
contains all but the rule defining type $\smash{\spec{t}^{(1)}}$. Clearly, at
this point a counter-example of $H\not\subseteq K$ exists, one whose root node
has type $\smash{\spec{t}^{(1)}}$ in $H$ but no type in $K$, however, it may be
small as it suffices to use a dag. To eliminate small counter-examples, by
adding them to the language of $K$, we ensure that all leaves of the counter-example 
are labeled with distinct subsets of $A$. In essence, we require in the
counter-example a node at level $i$ to have all leaves labeled with subset with
a set containing $\spec{a}_i$ and all leaves of its right subtree are labeled
with subsets missing $\spec{a}_i$. For that purpose, we introduce types
$\smash{\spec{s}^{(j)}_{i,\ONE,d}}$ ($\smash{\spec{s}^{(j)}_{i,\NONE,d}}$), 
which identify nodes at level $j$ that are
using (missing resp.) the symbol $\spec{a}_i$; 
the additional parameter $d\in\{\spec{L},\spec{R}\}$ is used to
handle disjunction and essentially indicates the subtree from which the usage
information comes from. The rules for leaves are (for $i\in\{1,\ldots,n\}$,
$M\in\{\NONE,\ONE\}$, and $d\in\{\spec{L},\spec{R}\}$)}
  &\smash{\spec{s}^{(n+1)}_{i,M,d}} \rightarrow
  \smash{\spec{a}_1\dbl{} \spec{t}_\spec{o}^\MAYBE \shuffle 
  \ldots\shuffle
  \spec{a}_{i-1}\dbl{} \spec{t}_\spec{o}^\MAYBE \shuffle{}} \\
  &\hspace{7em}\smash{\spec{a}_{i}\dbl{}\spec{t}_\spec{o}^M \shuffle
  \spec{a}_{i+1}\dbl{} \spec{t}_\spec{o}^\MAYBE \shuffle
  \ldots\shuffle
  \spec{a}_{n}\dbl{} \spec{t}_\spec{o}^\MAYBE}\\[-2pt]
\intertext{
The information of using a symbol $\spec{a}_i$ in a branch is propagated upward
but only to the level $i+1$ with the rules (for $i\in\{1,\ldots,n\}$,
$j\in\{i+1,\ldots,n\}$, and $M\in\{\NONE,\ONE\}$)}
&\smash{\spec{s}_{i,M,\spec{L}}^{(j)} 
    \rightarrow
    {\spec{L}\dbl{}\spec{s}_{i,M,\spec{L}}^{(j+1)\,\MAYBE}}
    \shuffle
    {\spec{L}\dbl{}\spec{s}_{i,M,\spec{R}}^{(j+1)\,\MAYBE}}
    \shuffle
    \spec{R}\dbl{}\spec{t}^{(j)}}\\[2pt]
&\smash{\spec{s}_{i,M,\spec{R}}^{(j)} 
    \rightarrow
    \spec{L}\dbl{}\spec{t}^{(j)}
    \shuffle
    {\spec{R}\dbl{}\spec{s}_{i,M,\spec{L}}^{(j+1)\,\MAYBE}}
    \shuffle
    {\spec{R}\dbl{}\spec{s}_{i,M,\spec{R}}^{(j+1)\,\MAYBE}}}\\[-2pt]
\intertext{
Finally, a tree is invalid for our purposes if a node at depth $i$ is missing
the symbol $\spec{a}_i$ in a leaf of its left subtree or is using the symbol
$\spec{a}_i$ in a leaf of its right subtree. This situation is identified and
propagated to the root node with the rules (for $i\in\{1,\ldots,n\}$ and
$j\in\{1,\ldots,i-1\}$)}
  &\smash{\spec{p}_{i,\spec{L}}^{(i)} 
    \rightarrow
    {\spec{L}\dbl{}\spec{s}_{i,\NONE,\spec{L}}^{(i+1)\,\MAYBE}}
    \shuffle
    {\spec{L}\dbl{}\spec{s}_{i,\NONE,\spec{R}}^{(i+1)\,\MAYBE}}
    \shuffle
    \spec{R}\dbl{}\spec{t}^{(i+1)}}\\[2pt]
  &\smash{\spec{p}_{i,\spec{R}}^{(i)} 
    \rightarrow
    \spec{L}\dbl{}\spec{t}^{(i+1)}
    \shuffle
    {\spec{R}\dbl{}\spec{s}_{i,\ONE,\spec{L}}^{(i+1)\,\MAYBE}}
    \shuffle
    {\spec{R}\dbl{}\spec{s}_{i,\ONE,\spec{R}}^{(i+1)\,\MAYBE}}}\\[2pt]
  &\smash{\spec{p}_{i,\spec{L}}^{(j)} 
    \rightarrow
    {\spec{L}\dbl{}\spec{p}_{i,\spec{L}}^{(j+1)\,\MAYBE}}
    \shuffle
    {\spec{L}\dbl{}\spec{p}_{i,\spec{R}}^{(j+1)\,\MAYBE}}
    \shuffle
    \spec{R}\dbl{}\spec{t}^{(j+1)}}\\[2pt]
  &\smash{\spec{p}_{i,\spec{R}}^{(j)} 
    \rightarrow
    \spec{L}\dbl{}\spec{t}^{(j+1)}
    \shuffle
    {\spec{R}\dbl{}\spec{p}_{i,\spec{L}}^{(j+1)\,\MAYBE}}
    \shuffle
    {\spec{R}\dbl{}\spec{p}_{i,\spec{R}}^{(j+1)\,\MAYBE}}}\\[-12pt]
\end{align*}
Now, the claim, proven with a simple induction, is that for any $G\in L(H)$
unless $G$ has an exponential tree, any node that has the type
$\smash{\spec{t}^{(1)}}$ of $H$ also has a type $\smash{\spec{p}_{i,d}^{(1)}}$
of $K$ for some $i\in\{1,\ldots,n\}$, $M\in\{\NONE,\ONE\}$, and
$d\in\{\spec{L},\spec{R}\}$. \qed
\end{proof}
The lower bound on the size of a minimal counter-example for \ShExZero is tight.
\begin{theorem}
  \label{thm:shex-zero-exp-example}
  For any $H,K\in\ShExZero$ such that $H\nsubseteq K$ there exists a graph
  $G\in L(H)\minus L(K)$ whose size is at most exponential in the size of $H$
  and $K$.
\end{theorem}
The proof consists of two parts. The first shows that there are at most
exponentially many kinds of nodes, and we use the same argument in the proof of
Theorem~\ref{thm:counter-example-shex} in Section~\ref{sec:shex}. The second
part uses a pumping argument to show that the outbound degree of a minimal
counter-example is polynomially bounded.
\subsection{Complexity}
The upper bound on the complexity of containment for \ShExZero is obtained with a
reduction from nondeterministic top-down tree automata known to be
\EXP-complete~\cite{tata07}. The reduction is not trivial because \RBEZero does
not allow directly the disjunction necessary to express nondeterminism, and the
proof need to account for graphs that do not represent trees.
\begin{theorem}
  \label{thm:shex-zero-pspace-hard}
  Containment for $\ShExZero$ is \EXP-hard.
\end{theorem}

The bound on the size of the counter example (Theorem~\ref{thm:shex-zero-exp-example}) yields immediately a \coNEXP upper bound. 
\begin{theorem}
  Containment for $\ShExZero$ is in \coNEXP.
\end{theorem}

\section{Shape Expression Schemas}
\label{sec:shex}
In this section, we address the question of decidability of containment for
\ShEx, which is far from obvious as \ShEx caries some expressive power of
$\exists\text{MSO}$ on graphs combined with Presburger
arithmetic\cite{SBLGHPS15} and monadic extensions of Presburger arithmetic
easily become undecidable~\cite{ElRa66}. We show that a triple-exponential upper
bound on the size of a counter-example, which can be compressed to
double-exponential size. The compression does not change the complexity of
validation, which permits us to give an upper bound on the complexity of testing
containment for \ShEx.


\subsection{Counter-example}

\paragraph{Compression}
Simple graphs do not allow multiple edges with the same label between the same
pair of nodes. We propose a model that allows it by attaching to every edge a
\emph{cardinality} indicating the number of such edges. More precisely, a
\emph{singleton} interval is an interval of the form $[k;k]$ for any natural
$k$, and a \emph{compressed} graph is a graph that uses only singleton intervals
on its edges and like simple graphs allows only one edge per label in $\Sigma$
between a pair of nodes. Given a compressed graph $\F$, its \emph{unpacking} is
a simple graph obtained by making a sufficient number of copies of each node,
each copy has the same outbound neighborhood but receiving at most one incoming
edge (cf. appendix). Since intervals are stored in binary, the unpacking of a
compressed graph $\F$ is of size at most exponential in the size of $\F$.
\begin{proposition}
  The size of the unpacking of a compressed graph $\F$ is at most exponential in
  the size of $\F$.
\end{proposition}
We adapt the definition of validation of \ShEx to compressed graphs by extending
the definition of node signature. Given a shape expression schema $S$ and a
compressed graph $\F$, the \emph{signature} of a node $n\in N_\F$ w.r.t.\ a
typing $T\subseteq N_\F\times \Gamma_S$ is
\[
  \sign_G^T(n)=
  \bigshuffle_{e\in\out_\F(n)}  
  \big(\big|_{t\in T(\target_\F(e))} \lab_\F(e)\dbl t\big)^{\occur_\F(e)}.
\]
Again, the typing $T$ is \emph{valid} iff
$L(\sign_G^T(n))\cap L(\delta_S(t))\neq\emptyset$ for every $(n,t)\in T$, there
exists a maximal valid typing of $\F$ w.r.t.\ $S$, and $\F$ satisfies $S$ if
$\dom(\Typing_{\F:S})=N_\F$. Naturally, if $\F$ satisfies $H$, then its
unpacking also satisfies $H$. Checking the satisfaction of \ShEx for compressed
graphs remains in \NP{} and to prove it we employ known results on Presburger
arithmetic that we present next.

\paragraph{Presburger Arithmetic}
The \emph{Presburger arithmetic} (PA) is the first-order logical theory of
natural numbers with addition that has decidable
satisfiability~\cite{Oppen78}. We point out that any natural number $n$ can be
easily defined with an existentially quantified formula of length linear in
$\log(n)$. Since we use PA formulas to define bags, we use a convenient
notation. When the set of symbols $\Delta=\{a_1,\ldots,a_k\}$ is known from the
context, a bag $w$ over $\Delta$ can be represented as a (Parikh) vector of $k$
natural numbers $\langle w(a_1), \ldots, w(a_k) \rangle$ and if a vector of
variables $\bar{x}$ is used to describe a bag over $\Delta$, we use elements of
$\Delta$ to index elements of $\bar{x}$: $x_a$ is $x_k$ for $k$ such that
$a_k=a$. Also, we we write $\varphi(w)$ to say that $\varphi$ is valid for $w$.

Because intersection is used to define satisfiability of a graph w.r.t.\ a
schema and intersection is easily expressed in Presburger arithmetic, we extend
\RBE with intersection $L(E_1\cap E_2)= L(E_1)\cap L(E_2)$. Now, for an regular
bag expression with intersection $E$ we recursively construct a formula
$\psi_E(x_1,\ldots,x_k, n)$ as follows.
\begin{small}
\begin{align*}
  \psi_\epsilon(\bar{x},n) 
  & {} \colonequals \textstyle\bigwedge_a x_a = 0
  \\
  \psi_{a}(\bar{x},n) 
  & {} \colonequals x_a = n \land 
    \textstyle\bigwedge_{b\neq a} x_b = 0
  \\
  \psi_{E^{[k;\ell]}}(\bar{x},n)
  & {} \colonequals (n=0 \land \textstyle\bigwedge_a x_a = 0) \lor{} \\
  & \phantom{{}\colonequals{}}(n > 0 \land \exists m.\ k\leq m \land m\leq\ell\land\psi_{E}(\bar{x},m))\\
  \psi_{E_1 \mid E_2}(\bar{x},n) 
  & {} \colonequals \exists \bar{x}_1, \bar{x}_2, n_1, n_2.\ 
    n = n_1 + n_2 \land{} \\ 
  & \phantom{{}\colonequals{}}\bar{x}=\bar{x}_1+\bar{x}_2 \land
    \psi_{E_1}(\bar{x}_1,n_1) \land \psi_{E_2}(\bar{x}_2,n_2) \\
  \psi_{E_1\shuffle E_2}(\bar{x},n)
  & {} \colonequals \exists \bar{x}_1, \bar{x}_2.\ 
    \bar{x} = \bar{x}_1 + \bar{x}_2 \land
    \psi_{E_1}(\bar{x}_1,n) \land{}\\
  & \phantom{{}\colonequals{}}\psi_{E_2}(\bar{x}_2,n)\\
  \psi_{E_1\cap E_2}(\bar{x},n) 
  & {}\colonequals \psi_{E_1}(\bar{x},n) \land \psi_{E_2}(\bar{x},n) 
\end{align*}%
\end{small}%
The main claim, proven with a simple induction, is that $\psi_E(w, n)$ iff
$w\in L(E)^n$ for any bag $w$ over $\Delta$ and any $n\geq 0$. It follows that
$L(E)\neq\emptyset$ iff $\exists \bar{x}. \psi_E(\bar{x},1)$ is valid. Validity
of existentially quantified PA formulas is known to be in NP~\cite{Gradel87},
and consequently, we obtain an upper bound on complexity of validation of
compressed graphs.
\begin{proposition}
\label{prop:validation-compressed-graphs-shex}
Validation of compressed graphs w.r.t.\ \ShEx is in \NP.   
\end{proposition}
The following result is instrumental in our analysis of upper bounds of a
counter-example for \ShEx.
\begin{proposition}[\cite{Weis90}]
  \label{prop:pa-bound}
  Let $\Phi=Q_1\bar{x}_1\ldots Q_k\bar{x}_k.\varphi$ be a closed formula of
  Presburger arithmetic in prenex normal form with $k$ quantifier alternations
  over the variables $\bar{x}=\bar{x}_1\cup\ldots\cup\bar{x}_k$ ($\varphi$ is
  quantifier-free). Then $\Phi$ is valid if and only if $\Phi$ is valid when
  restricting the first-order variables of $\Phi$ to be interpreted over
  elements of $\{0,\ldots,B\}$, where $\log(B)=\smash{O(|\varphi|^{3|\bar{x}|^k})}$.
\end{proposition}

\paragraph{Compressed counter-example}
We now assume two schemas $H$ and $K$ such that $H\nsubseteq K$ and take any
counter-example $G\in L(H)\minus L(K)$. We know that there is at least one node
of that satisfies at least one type of $H$ but no type of $K$. In general, for a
node $n$ of $G$ we identify a pair $(T,S)$ consisting of a set $T$ of types of
$H$ and a set of types $S$ of $K$ that $n$ satisfies. We say that the node $n$
is of the \emph{$(T,S)$-kind} and we identify the set $\C$ of all kinds present
in $G$.
\begin{align*}
  &\kind(n) = (\Typing_{G:H}(n),\Typing_{G:K}(n)),\\
  &\C = \{\kind(n) \mid n\in N_G\}.
\end{align*}
Shape expression schemas may only inspect the labels of the outgoing edges of a
node and the types of the nodes at the end points of the edges. Consequently, if
we replace a node by a node of the same kind, or more precisely we redirect all
incoming edges of the first node to the other node, then the types of node in
the graph do not change, in particular, it remains a
counter-example. Furthermore, we can fuse the set of all nodes of the same class
into a single node that belongs to the same class, and still obtain a graph that
is a counter-example. When fusing several nodes we gather the incoming edges
into a fused node but for the outgoing edges we use only the outgoing edges of a
one (arbitrarily chosen) of the fused nodes, while discarding the outgoing edges
of the remaining nodes. We point out that the obtained graph needs not longer to
be simple, fusing a set of nodes may lead to several incoming edges with the
same label originating from the same node. Such multiple edges can, however, be
easily compressed to a single one.

We describe the construction of the compressed counter-example $\F$ more
precisely. First for every kind $\kappa\in\C(G)$ we pick an (arbitrarily chosen)
representative node $n_\kappa\in G$ such that $\kind(n_\kappa) = \kappa$. The
set of nodes of $\F$ is the set of all classes of $G$, $N_\F=\C(G)$. For every
edge connecting two representative nodes $\F$ has a corresponding edge:
\begin{multline*}
E_\F = \{ \langle \kappa,a,\kappa'\rangle \mid 
\exists e\in E_G.\ 
\source_G(e) = n_C,\\
\target_G(e) = n_{C'},\ 
\lab_G(e) = a
\}
\end{multline*}
and for $\langle\kappa,a,\kappa'\rangle\in E_\F$
\begin{align*}
  &\source_\F(\langle\kappa,a,\kappa'\rangle) = \kappa,&
  &\lab_\F(\langle\kappa,a,\kappa'\rangle) = a,\\
  &\target_\F(\langle\kappa,a,\kappa'\rangle) = \kappa',&
  &\occur_\F(\langle\kappa,a,\kappa'\rangle) = [k;k],
\end{align*}
where
\begin{small}
\[
  k = |\{e\in \out_G(n_\kappa) \mid \lab_G(e)=a,\ \kind(\target_G(e)) = \kappa'\}|.
\]
\end{small}%
This shows that a compressed counter-example needs to have at most an
exponential number of nodes. 

\paragraph{Bounding the node degree}
The remaining question is how big the cardinalities of the edges of $\F$ must
be. We answer this question with the help of Proposition~\ref{prop:pa-bound} by
describing the outbound neighborhood of a node of $\F$ with Presburger
arithmetic formula. 

For the kind $(T,S)\in\C$ the formula $\Phi_{(T,S)}$ examines the existence of
an outbound neighborhood of a node of that kind that satisfies all types in $T$
and all types in $S$. This neighborhood is captured as a bag $\bar{x}$ over
$\Delta_\C=\{a\dbl(T',S') \mid a\in\Sigma,\ (T',S')\in\C\}$, where an occurrence
of the symbol $a\dbl{}(T',S')$ corresponds to one outgoing edge labeled with $a$
and leading to a node of the kind $(T',S')$.
\begin{align*}
  \Phi_{(T,S)} \colonequals \exists \bar{x}.\ 
  &
    \bigwedge_{t\in T} \varphi_t(\bar{x}) \land 
    \bigwedge_{\mathclap{t\in \Gamma_H\minus T}} \neg \varphi_t(\bar{x}) \land{}\\
  &
    \bigwedge_{s\in S} \varphi_s(\bar{x}) \land 
    \bigwedge_{\mathclap{s\in \Gamma_K\minus S}} \neg \varphi_s(\bar{x}).
\end{align*}
The formulas $\varphi_t(\bar{x})$ and $\varphi_s(\bar{x})$ verify whether the
types $t$ of $H$ and $s$ of $K$ are satisfied in this neighborhood. This is done
in two phases and we present it only for $\varphi_t(\bar{x})$; the formula
$\varphi_s(\bar{x})$ is defined analogously. The variable $x_{a\dbl{}(T',S')}$
represents the number of outgoing edges with label $a$ to nodes that satisfy all
types in $T'$. In the context of satisfying definition of the type $t$ each
outgoing edge is used with exactly one type. Consequently, the next formula
partitions the number of  outgoing edges $x_{a\dbl{}(T',S')}$ into all types
in $T'$. Here, we use a vector $\bar{y}$ of variables over
$\{a\dbl(T',S')\rightarrow a\dbl t' \mid a\in\Sigma,\ (T',S')\in\C,\ t'\in
T'\}$, where $y_{a\dbl(T',S')\rightarrow a\dbl t'}$ represents the part of
$x_{a\dbl{}(T',S')}$ edges that is to be used with the type $t'$.
\begin{small}
\begin{align*}
  &\varphi_t(\bar{x})\colonequals
    \exists \bar{y}.\ 
    \bigwedge_{\mathclap{a\dbl{}(T',S')\in\Delta_\C}}
    x_{a\dbl{}(T',S')}=\sum_{t'\in T'} y_{a\dbl{}(T',S')\rightarrow a\dbl t'}
    \land \varphi_t'(\bar{y}).
\end{align*}
\end{small}%
Finally, the edges with the same label and type of the end point are aggregated
in the vector $\bar{z}$ representing a bag over $\Delta_H=\Sigma\times\Gamma_H$,
which is then is fed to the formula $\psi_{\delta_H(t)}$ that defines the type
definition of $t$.
\begin{small}
\begin{align*}
  &\varphi_t'(\bar{y}) \colonequals 
    \exists \bar{z}.\ 
    \bigwedge_{\mathclap{a\dbl{} t'\in\Delta_H}}
    z_{a\dbl t'} = 
    \sum_{\mathclap{\substack{a\dbl(T',S')\in\Delta_\C\\\text{s.t. $t'\in T'$}}}} 
  y_{a\dbl(T',S')\rightarrow a\dbl t'} 
    \land\psi_{\delta_H(t)}(\bar{z}).
\end{align*}
\end{small}%
The formula $\Phi_{(T,S)}$ can be easily converted to prenex normal form, and
then, it is of exponential length, uses an exponential number of quantified
variables, and has only one alternation of quantifiers. Since $\Phi_{(T,S)}$ is
valid for any $(T,S)\in\C$, the satisfying values for the variables $\bar{x}$
are bound by a triple exponential, and consequently, have a binary
representation whose size is bounded by double-exponential function in the size
of $H$ and $K$.
\begin{theorem}
  \label{thm:counter-example-shex}
  For any two \ShEx $H$ and $K$, if $H \nsubseteq K$, then there exists a
  compressed graph $\F$ that satisfies $H$, does not satisfy $K$, and whose size
  is at most double-exponential in the size of $H$ and $K$.
\end{theorem}

\subsection{Complexity} 
Very recently containment for \RBE has been shown to be
\coNEXP-complete~\cite{HaHo16}, and consequently, we obtain the following lower
bound.
\begin{proposition}
  Containment for \ShEx is \coNEXP-hard.
\end{proposition}
The upper bound follows from Theorem~\ref{thm:counter-example-shex} and
Proposition~\ref{prop:validation-compressed-graphs-shex}. A (universally)
nondeterministic Turing for an input pair $(H,K)$ guesses a compressed graph
$\F$ and uses an \NP\ oracle to verify that if $\F$ satisfies the schema $H$,
the it also satisfies the schema $K$. The input pair is accepted if the test is
passed on every computation path.
\begin{corollary}
  Containment for \ShEx is in $\coTwoNEXP^\NP$.
\end{corollary}

\section{Conclusions and Future Work}
\label{sec:conclusions}
This work has been prompted by our current work on data exchange for RDF and
schema reference for RDF, where not only do we ask the questions of type
implication but are also interested in instances satisfying constraints
expressed with the help of \ShEx. In this paper, we have considered \ShEx and
its two practical subclasses \ShExZero and \DetShExZeroMinus. While the precise
complexity of containment for \ShEx remains open, the complexity results we have
obtained, summarized in Figure~\ref{fig:summary}, provide a reasonably good
separation of the classes of schemas.
\begin{figure}[htb]
  \centering
  \renewcommand{\arraystretch}{1.5}
  \begin{tabular}{|c|c|c|}
    \hline
    \DetShExZeroMinus
    &
    \ShExZero
    &
    \ShEx\\
    \hline
    \multirow{2}{*}{\PTIME}
    &
      \EXP-hard
    &
      \coNEXP-hard
    \\
    &
      \coNEXP
    &
      $\coTwoNEXP^\NP$
    \\
    \hline
  \end{tabular}
  \caption{Summary of complexity results}
  \label{fig:summary}
\end{figure}
While determinism shows promises to allow reduction in complexity. For instance,
containment for \DetShEx is in \coTwoNEXP\ since validation for \DetShEx is in
\PTIME. But its precise impact on complexity of containment needs to be studied
further. the class of regular bag expression \textsf{DIME} that permits
restricted use of disjunction yet allows for tractable containment for schemas
for unordered XML~\cite{BoCiSt14j} and it would also be interesting to see if
there are any computational benefits that can be drawn for shape expression
schemas using \textsf{DIME}.


\clearpage
\onecolumn

\bibliographystyle{plain}
\bibliography{staworko}

\end{document}